\documentclass[a4paper,UKenglish,cleveref, autoref, thm-restate]{lipics-v2021}

\usepackage[noend]{algpseudocode}
\usepackage[]{algorithm2e}

\usepackage[utf8]{inputenc}
\usepackage{epsfig}
\usepackage{amsmath}
\usepackage{amsthm}
\usepackage{amsfonts}
\usepackage{amssymb}
\usepackage{listings}
\usepackage{mathtools}
\usepackage{bm}
\usepackage{graphicx}
\usepackage{pdfpages}
\usepackage{float}
\usepackage{thmtools}
\usepackage{thm-restate}
\usepackage{tikz}
\usepackage{color}
\usepackage{hyperref}

\newcommand{\eps}{\varepsilon}
\newcommand{\area}{\mathrm{area}}
\newcommand{\gravity}{\textsc{Gravity}}
\newcommand{\exchange}{\textsc{Exchange}}
\newcommand{\map}{M}

\usepackage{inputenc,graphicx}

\usepackage{enumitem}
\newlist{Aenumerate}{enumerate}{1}
\setlist[Aenumerate]{label=A.\arabic*}

\title{Reconfiguration of Polygonal Subdivisions via Recombination}

\author{Hugo A. Akitaya}{University of Massachusetts Lowell, MA, USA}{hugo_akitaya@uml.edu}{https://orcid.org/0000-0002-6827-2200}{}
\author{Andrei Gonczi}{Tufts University, MA, USA}{andrei.gonczi@tufts.edu}{https://orcid.org/0000-0002-5939-2366}{}
\author{Diane L. Souvaine}{Tufts University, MA, USA}{andrei.gonczi@tufts.edu}{ORCID}{}
\author{Csaba D. T\'oth}{California State University Northridge, Los Angeles, CA, USA \and Tufts University, Medford, MA, USA}{csaba.toth@csun.edu}{https://orcid.org/0000-0002-8769-3190}{Research funded in part by NSF awards MS-1800734 and DMS-2154347}
\author{Thomas Weighill}{University of North Carolina Greensboro, NC, USA}{t_weighill@uncg.edu}{https://orcid.org/0000-0003-2614-0979}{Research funded in part by NSF award OIA-1937095.}

\authorrunning{H.~A.~Akitaya et al.} 

\Copyright{Hugo A. Akitaya,, Andrei Gonczi, Diane L. Souvaine, Csaba D. T\'oth, and Thomas Weighill} 

\ccsdesc[500]{Social and professional topics~Social engineering attacks}
\ccsdesc[500]{Theory of computation~Computational geometry}
\ccsdesc[500]{Theory of computation~Nonconvex optimization}

\keywords{configuration space, gerrymandering, polygonal subdivision, recombination} 

\category{} 

\relatedversion{} 

\nolinenumbers 

\hideLIPIcs

\begin{document}
\maketitle

\begin{abstract}
    Motivated by the problem of redistricting, we study area-preserving reconfigurations of connected subdivisions of a simple polygon.
    A connected subdivision of a polygon $\mathcal{R}$, called a \emph{district map}, is a set of interior disjoint connected polygons called \emph{districts} whose union equals $\mathcal{R}$. 
    We consider the \emph{recombination} as the reconfiguration move which takes a subdivision and produces another by merging two adjacent districts, and by splitting them into two connected polygons of the same area as the original districts.
    The \emph{complexity} of a map is the number of vertices in the boundaries of its districts.
    Given two maps with $k$ districts, with complexity $O(n)$, and a perfect matching between districts of the same area in the two maps, we show constructively that $(\log n)^{O(\log k)}$ recombination moves are sufficient to reconfigure one into the other.
    We also show that $\Omega(\log n)$ recombination moves are sometimes necessary even when $k=3$, thus providing a tight bound when $k=O(1)$.
\end{abstract}

\section{Introduction}
\label{sec:intrp}

We consider the problem of \emph{redistricting}---the partition of a geographic domain into disjoint districts. In particular, we consider the case when these districts are required to be connected and of roughly equal population. These criteria are typically enforced in political redistricting, wherein each district elects one or more representatives to serve on a governing body, a canonical example being Congressional districts in the United States. Even under these restrictions, the space of possible redistricting plans for a typical domain 
is intractably vast, making it difficult to sample from this space. Recently, algorithms for generating large samples of plans have made it possible to establish a neutral baseline for a particular state which in turn can be used to detect and describe gerrymanders (i.e., unfair maps)~\cite{chen2015cutting,chen2013unintentional, chikina2017assessing,deford2019redistricting, DeFord2021Recombination, Mattingly2018,herschlag2017evaluating}.  

The most common and successful sampling algorithms for redistricting are Markov chains that perform a sequence of reconfiguration moves on an initial map. The most prominent reconfiguration move is the \emph{recombination} or \emph{ReCom} move (see Figure~\ref{fig:WI}), which is a move that modifies two adjacent districts while maintaining population balance and connectivity~\cite{deford2019redistricting, DeFord2021Recombination}. In order to properly sample from the space of redistricting plans, we should require that any feasible redistricting plan can be reached from the initial map by a finite sequence of ReCom moves. That is, we want to positively answer the \emph{reachability} question for this reconfiguration move; in the language of Markov chains this would be to prove that any chain built on the ReCom move is \emph{irreducibile}. 


Historically, most redistricting algorithms have operated on a discretized version of the geographic domain. In this framework, a district map is modeled as a vertex partition of an adjacency graph~\cite{DeFord2021Recombination,duchinAAAS}. This is natural since population data is only available at the level of fixed geographic units, such as Census blocks in the case of the United States. The ReCom algorithm fits within this framework, and current versions all use a spanning tree method on the adjacency graph to perform the ReCom move. Unfortunately, it is easy to construct small pathological examples of graphs for which ReCom reachability fails. Moreover, even determining whether two plans can be connected via a sequence of ReCom moves is PSPACE-complete~\cite{akitaya2022reconfiguration} for general (planar) graphs.

A reasonable but unproven hypothesis is that for real-world adjacency graphs representing sufficiently fine discretizations of the geographic domain, we will indeed have reachability. A general theorem covering all adjacency graphs of interest seems beyond reach, which has led to a search for intermediate results. One direction of investigation is to allow a large class of graphs but relax the population balance constraint considerably; in such cases theoretical results are possible~\cite{akitaya2019reconfiguration, akitaya2022reconfiguration} (see Related Work below). Reachability on grid graphs or triangular lattices is an active area of research but as of yet without concrete results. 

In this paper, we return to the original hypothesis---that sufficient discretization leads to reachability---to motivate our result. Instead of modeling redistricting plans as graph partitions, we adopt a continuous model where the districts are connected polygons of equal population which partition a polygonal domain. Note that sampling algorithms based on this model do exist in the literature, most notably the power diagram method in \cite{cohen2018balanced}, but these algorithms are not Markov chains and require an extra refinement step to go from polygonal partitions to partitions that respect the geographic units. 

In our continuous model, we are able to establish reachability for the ReCom move---that is, any two polygonal partitions can be connected by finitely many ReCom steps that merge and resplit adjacent polygonal districts. The implication is that given two real-world redistricting plans, a sufficiently fine discretization of the geographic domain allows a finite sequence of ReCom moves (on the adjacency graph) to connect them. In practice this could mean that a particular map is not reachable from the initial map when considering voting precincts as geographic units, but could become reachable when working with Census blocks.

\subparagraph{Related Work.} In the discrete setting, the context for the reachability problem consists of a graph $G$ with $n$ nodes, a number of districts $k$ and a \emph{slack} $\varepsilon \ge 0$. Valid partitions are defined as partitions of $V(G)$ into $k$ non-empty subsets (called \emph{districts}) that each induce connected subgraphs such that the number of vertices in each district lies in the interval $[(1-\varepsilon) \frac{n}{k}, (1+\varepsilon)\frac{n}{k}]$. Two common reconfiguration moves on the space of valid partitions are the switch move and ReCom move. A \emph{switch move}~\cite{fifield2020automated, mattingly2014redistricting} consists of reassigning a single node to a new district. 
Using the switch move 
allows one to construct a Markov chain on the space of valid partitions with easily computable transition probabilities. A Metropolis-Hastings weighting can then be used to ensure that the chain samples (in the limit) from any desired distribution on the space of valid partitions. Crucially, however, this relies on the assumption that the state space is connected, i.e., that any two partitions can be connected by switch moves. It is not hard to design concrete examples of graphs for which this is not true with $\varepsilon = 0$. It is known that for $\varepsilon = \infty$, the state space is connected under the switch move when $G$ is biconnected; furthermore, that deciding whether two partitions can be connected by switch moves is PSPACE-complete even when $G$ is planar~\cite{akitaya2019reconfiguration}. 

The usefulness of the switch move is hampered by the fact that Markov chains built using it tend to mix slowly \cite{najt2019complexity}. As a result, larger reconfiguration moves, that are often more effective on real-world instances, were introduced. The \emph{ReCom move}~\cite{deford2019redistricting, DeFord2021Recombination} consists of merging and resplitting two adjacent districts (note that the switch move is a special case of a ReCom move). When designing a Markov chain based on this move, the most common method for resplitting is to draw a random spanning tree of the merged districts and cut an edge such that the resulting connected components form a valid partition. 
The disadvantage to such a process is that the transition probabilities between partitions appear to be intractable, so that the resulting Markov chain has an unknown stationary distribution. Recently, modifications of the original ReCom Markov chain have been proposed which have computable transition probabilities~\cite{autry2020multi, cannon2022spanning}; however, an accurate description of the 
the stationary distribution still requires the state space to be connected. It is easy to construct a graph $G$ for which the space of valid partitions is not ReCom-connected for $\varepsilon = 0$ (even for a 6$\times$6 grid graph~\cite{cannon2022spanning}). It is known~\cite{akitaya2022reconfiguration} that the state space is connected whenever $G$ is connected  and $\varepsilon = \infty$, and also when $G$ is  Hamiltonian and $\varepsilon \geq 2$; furthermore deciding whether two partitions can be connected by ReCom moves is PSPACE-complete even when $G$ is a triangulation.

\begin{figure}
    \centering%
    \begin{tikzpicture}[scale = 0.875]
    \node at (0,0) {\includegraphics[width=0.22\textwidth]{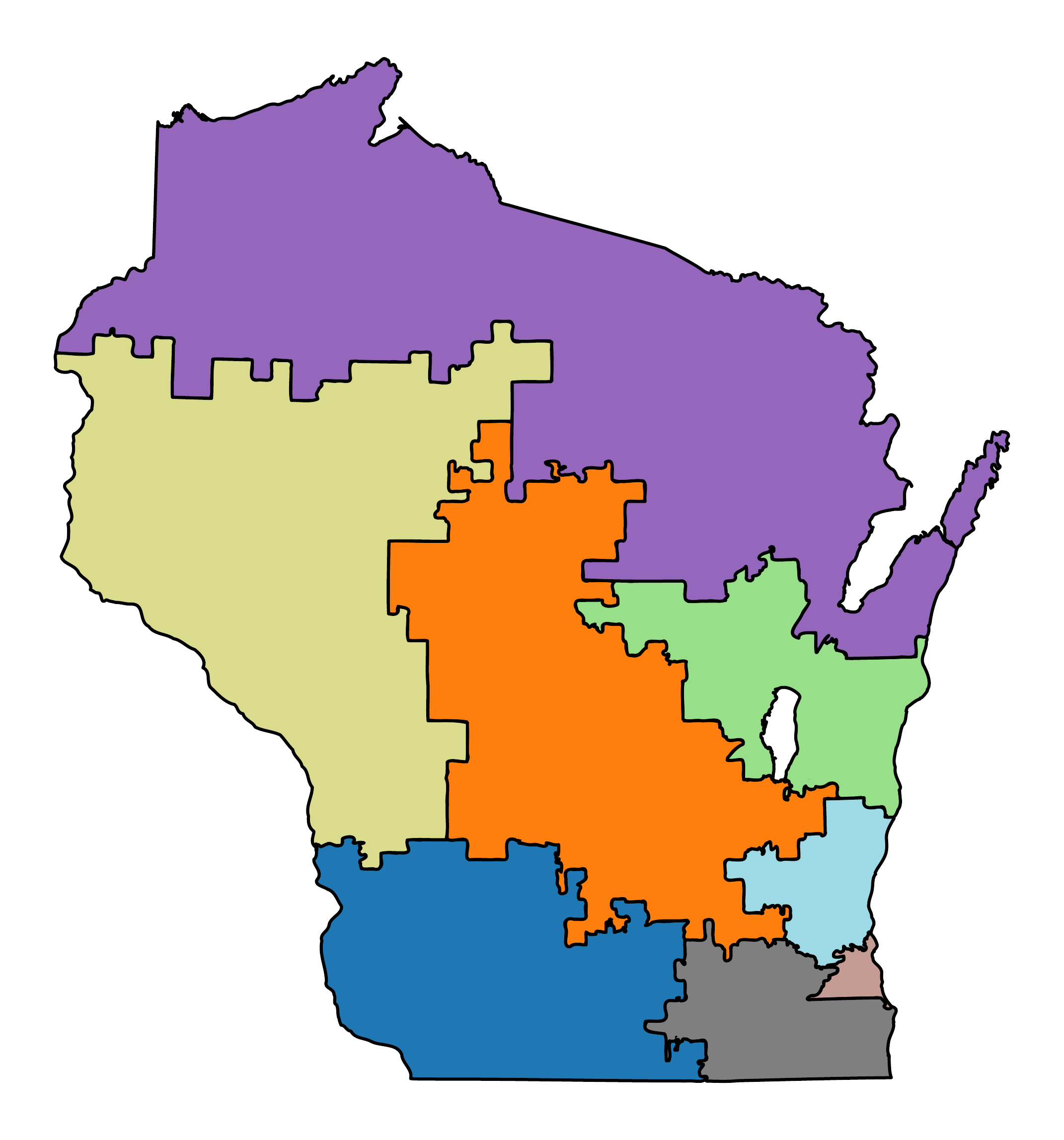}};
    \node at (4,0) {\includegraphics[width=0.22\textwidth]{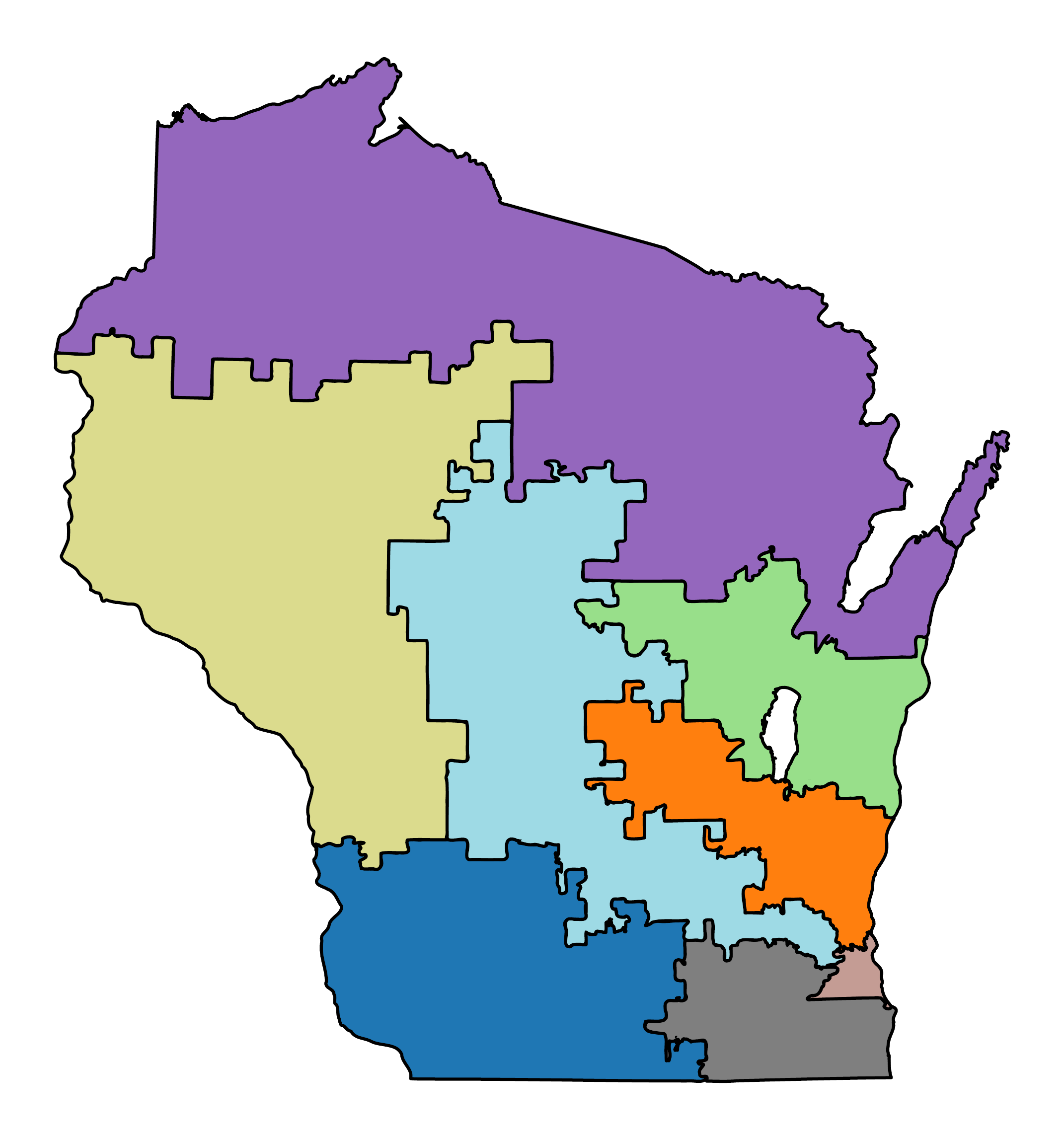}};
    \node at (8,0) {\includegraphics[width=0.22\textwidth]{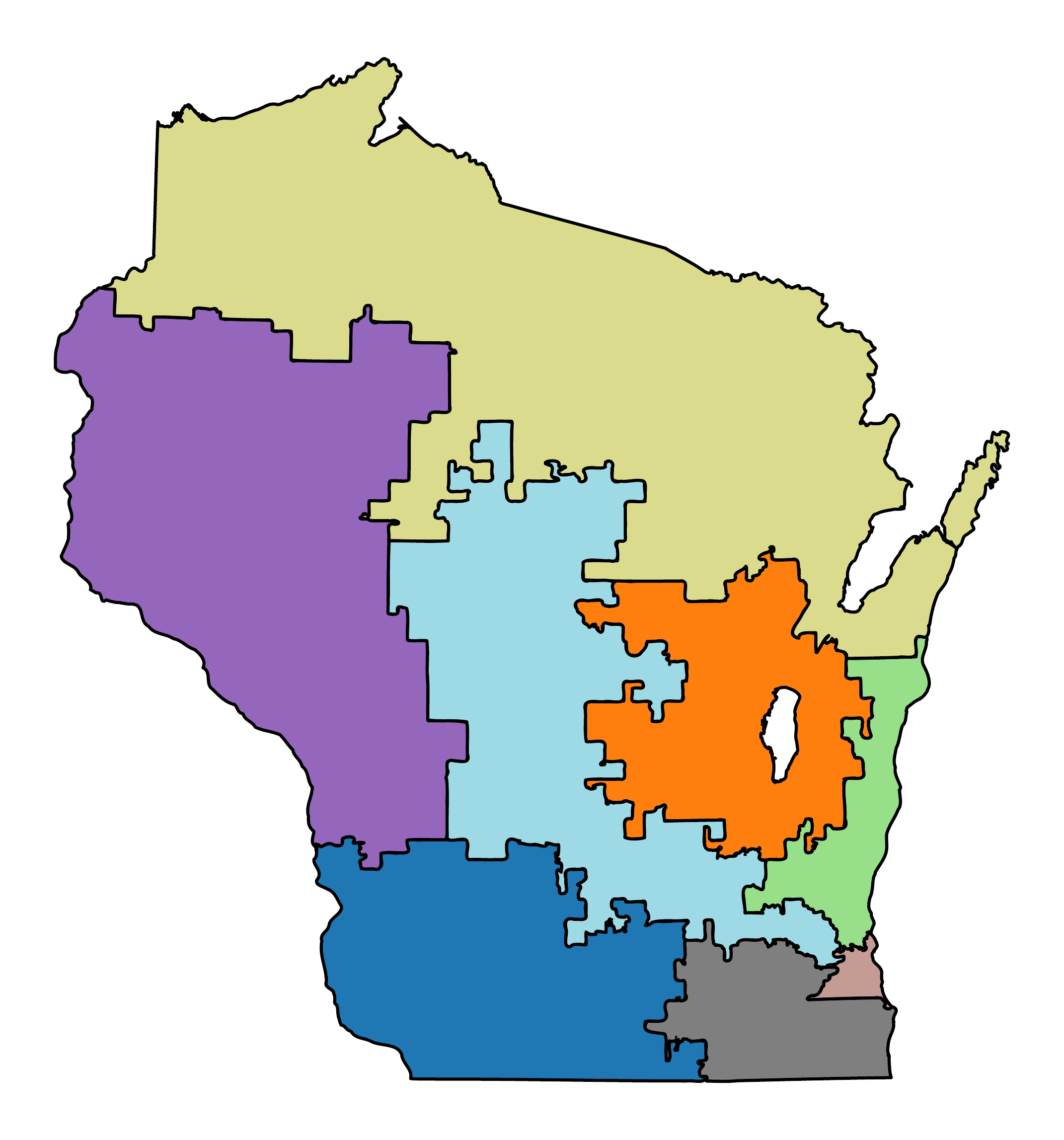}};
    \node at (12,0) {\includegraphics[width=0.22\textwidth]{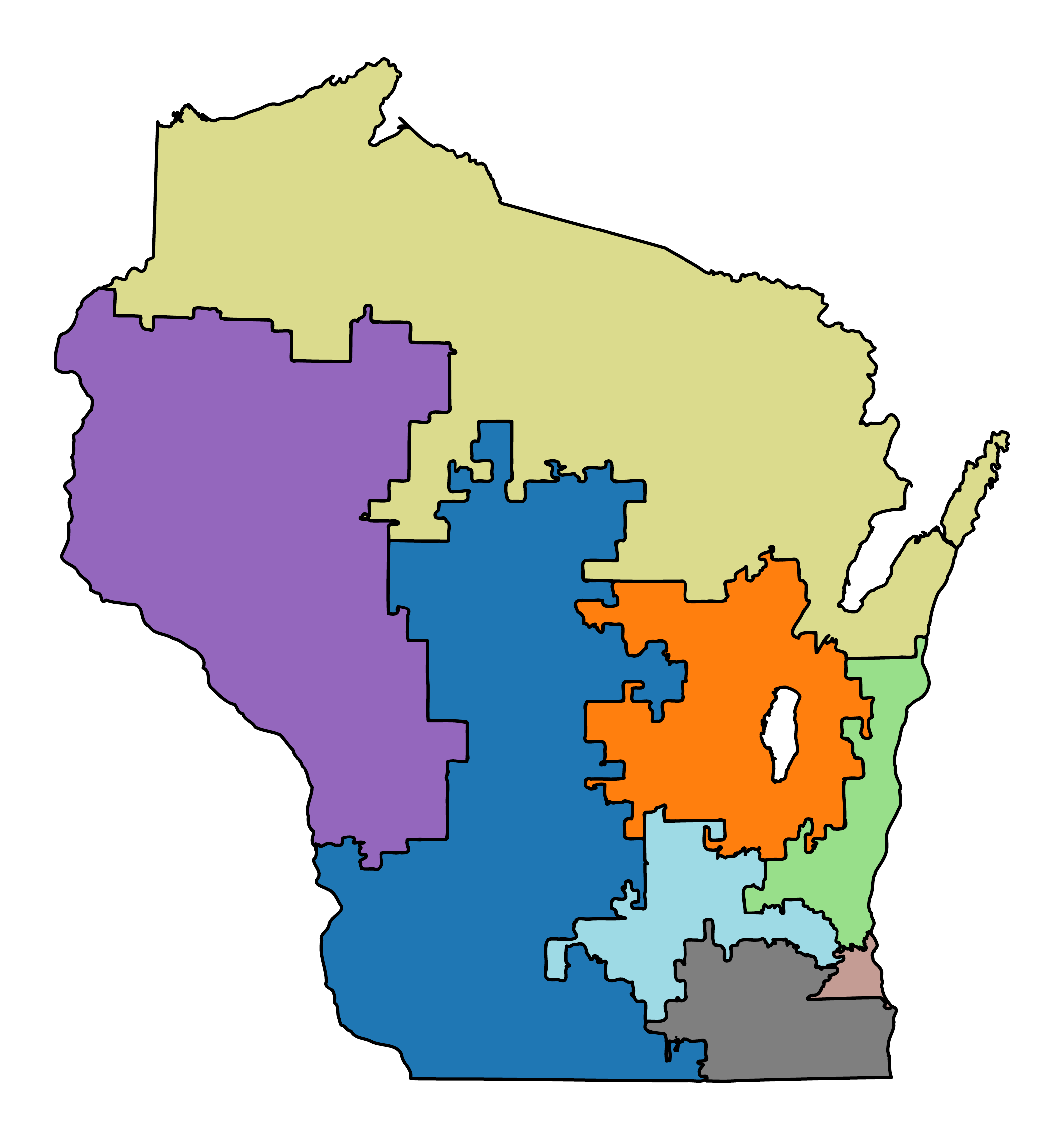}};
    \draw[->, line width=1mm] (1.75,-0.25)--(2.25,-0.25);
    \draw[->, line width=1mm] (5.75,-0.25)--(6.25,-0.25);
    \draw[->, line width=1mm] (9.75,-0.25)--(10.25,-0.25);
    \end{tikzpicture}
    \caption{A sequence of three recombination moves on the state of Wisconsin. At each step, two districts are merged and split again. The reachability problem is to determine whether any map can be reached from any other by a finite sequence of such steps.}
    \label{fig:WI}
\end{figure}
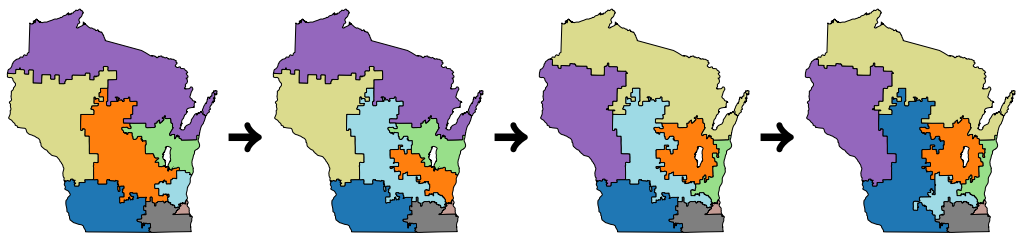

\smallskip\noindent{\sffamily\normalsize\bfseries Contributions.}
In this paper, we introduce a continuous model for redistricting and ReCom moves, where the districts can be arbitrary connected polygons (with real coordinates) in a polygonal domain (Section~\ref{sec:pre}). 
While the configuration space in this setting contains infinitely many maps, we prove that it is always connected under ReCom moves. Our proof is constructive, and provides an upper bound on the minimum number of ReCom moves between any two maps in terms of the number of districts $k$ and the complexity of the district maps $n$ (i.e., the number of vertices of all polygons in the initial and target maps). We start with the first nontrivial case, $k=3$  districts in a unit square domain, and show that between any two maps of complexity $O(n)$, there is a reconfiguration path consisting of $O(\log n)$ ReCom moves (Theorem~\ref{thm:square3log} in Section~\ref{sec:three}). Importantly, the complexity of the map remains $O(n)$ in all intermediate steps. Our reconfiguration algorithm generalizes to $k$ districts in an arbitrary polygonal domain, using a recursion of depth $O(\log k)$.
It yields an $\exp(O(\log k\, \log \log n))$ bound on the number of ReCom moves between two maps; however, for the complexity of intermediate maps we obtain only a weaker bound of $n^{k^{O(1)}}$  (Theorem~\ref{thm:kdistrict} in Section~\ref{sec:general}). On the other hand, we show that (even for $k=3$) the diameter of the configuration space is infinite by constructing pairs of maps which require arbitrarily large numbers of ReCom moves to connect (Theorem~\ref{thm:lower-bound} in Section~\ref{sec:lower}). The number of moves for these examples grows logarithmically with the complexity of the maps, thereby providing a lower bound which perfectly matches our upper bound.

\section{Preliminaries}
\label{sec:pre}

A \emph{region} is a connected set in $\mathbb{R}^2$ bounded by one or more pairwise disjoint Jordan curves. A \emph{$k$-district map} $\map(\mathcal{R}) = \{D_1,\ldots,D_k\}$ is a decomposition of a region $\mathcal{R}$ into $k$ interior-disjoint regions (that is, $\mathcal{R}=\bigcup_{i=1}^k D_i$ and $\text{int}(D_i)\cap \text{int}(D_j)=\emptyset$ for $i \neq j$), where $\mathcal{R}$ is the \emph{domain}, and $D_1,\ldots , D_k$ are the \emph{districts} of the map. We may refer to $\map(\mathcal{R})$ simply as $\map$ if $\mathcal{R}$ is clear from the context.  
A \emph{recombination} move (for short, \emph{ReCom}) takes a map $\map(\mathcal{R})$ and two 
districts $D_i, D_j \in \map(\mathcal{R})$ and returns a new district map of the same domain  $\map'(\mathcal{R})=\map(\mathcal{R})\setminus \{D_i,D_j\}\cup \{D_i',D_j'\}$. 
A recombination is \emph{area-preserving} if $\area(D_i) = \area(D_i')$ and $\area(D_j) = \area(D_j')$. 
Two $k$-district maps, $\map(\mathcal{R})=\{D_1,\ldots,D_k\}$ and $\map'(\mathcal{R})=\{D'_1,\ldots,D'_k\}$, on a domain $\mathcal{R}$ are \emph{area-compatible} if there is a permutation $\pi:\{1,\ldots ,k\}\rightarrow \{1,\ldots ,k\}$ such that  $\area(D_i) = \area(D'_{\pi(i)})$ for all $i=1,\ldots ,k$.

Here, we assume that the domain $\mathcal{R}$ is a simple polygon;
and each district is a connected polygon (possibly with holes).
The \emph{configuration space} of a map $M(\mathcal{R})$ is the set of all polygonal district maps on $\mathcal{R}$ that are area-compatible with $M(\mathcal{R})$.
We define the \emph{complexity} of a map $\map$ as the total number of vertices on the boundaries of all districts in $\map(\mathcal{R})$.
We show (in Section~\ref{sec:general}) that w.l.o.g.\ we may assume a unit square domain $\mathcal{R}=[0,1]^2$. 
The \emph{area} of a polygon $P$, denoted $\area(P)$, is either the Euclidean area of $P$ or the integral $\int_P \delta$ of some nonnegative integrable density function $\delta:\mathcal{R}\rightarrow \mathbb{R}_{\geq 0}$.
We show (by Theorem~\ref{thm:kdistrict}) that any pair of area-compatible district maps can be reconfigured into each other by a sequence of area-preserving recombinations (i.e., the configuration space of area-compatible district maps is connected).

\subparagraph{Weak Representation.} 
In intermediate maps of a ReCom sequence, we use infinitesimally narrow \emph{corridors} to keep the districts connected. In order to handle narrow corridors efficiently, we rely on a compressed representation of district maps using \emph{weak embeddings} (defined below), where each corridor is represented by a polygonal path; see Fig.~\ref{fig:thickening}. The compressed representation has two key advantages: 
(1) We may assume that corridors have zero area; and (2) 
we may reduce the total number of vertices by representing several parallel corridors by overlapping polygonal paths (with shared vertices). In Sections~\ref{sec:three}--\ref{sec:general}, we construct a sequence of ReCom moves on compressed maps. We show (in Proposition~\ref{cor:perturbation} below) that the polygonal paths can be thickened into narrow corridors in each stage of the ReCom sequence to produce a ReCom sequence in which the districts are simple polygons. 
%

An \emph{embedding} of a planar graph $G$ is an injective function from $G$ (seen as a 1-dimensional topological space) to $\mathbb{R}^2$; intuitively it is a drawing of $G$ in which edges can intersect only at common endpoints.
A \emph{weak embedding} of $G$ is a continuous function from $G$ to $\mathbb{R}^2$ such that, for every $\eps>0$, each vertex can be moved by at most $\eps$ and each edge can be replaced by a curve within Fr\'echet distance $\eps$ to form an embedding of $G$ (i.e., an \emph{$\eps$-perturbation} of a weak embedding is an embedding).
In particular, a \emph{simple polygon} is a piecewise linear embedding of a cycle and the region bound by it; and a \emph{weakly simple polygon} is a piecewise linear weak embedding of a cycle and the region bounded by it.
A \emph{polygon} (with possible holes) is a simple polygon with pairwise disjoint simple polygons (holes) removed. Similarly, a \emph{weak polygon} is a weakly simple polygon with pairwise disjoint weakly simple polygons removed.

\begin{figure}[tbh]
	\centering
    \begin{minipage}{.45\textwidth}
          \centering
          \includegraphics[width = 0.8 \textwidth]{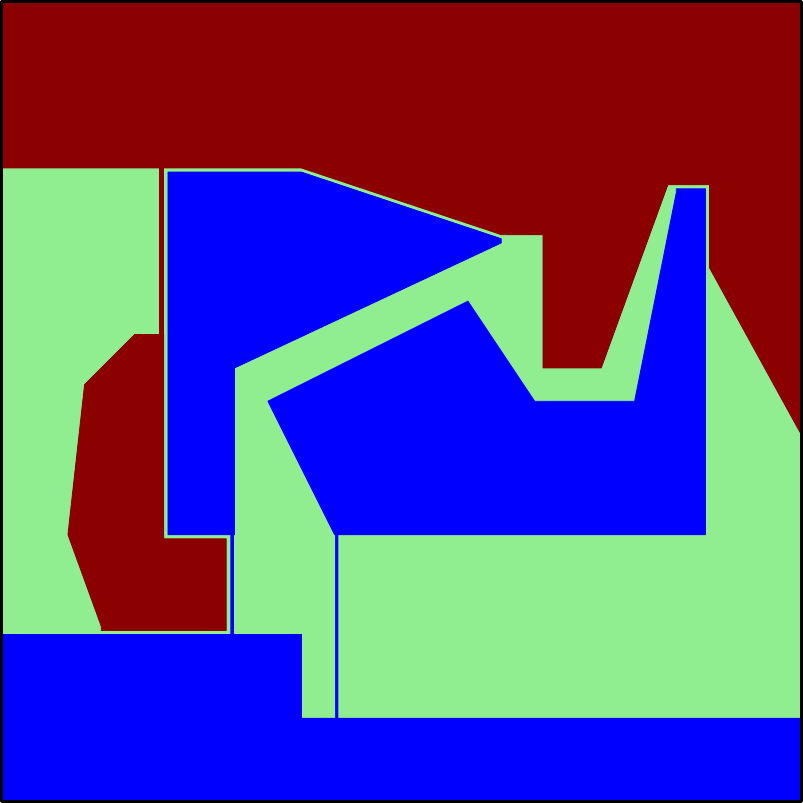}
    \end{minipage}
    \begin{minipage}{.45\textwidth}
          \centering
          \includegraphics[width = 0.8 \textwidth]{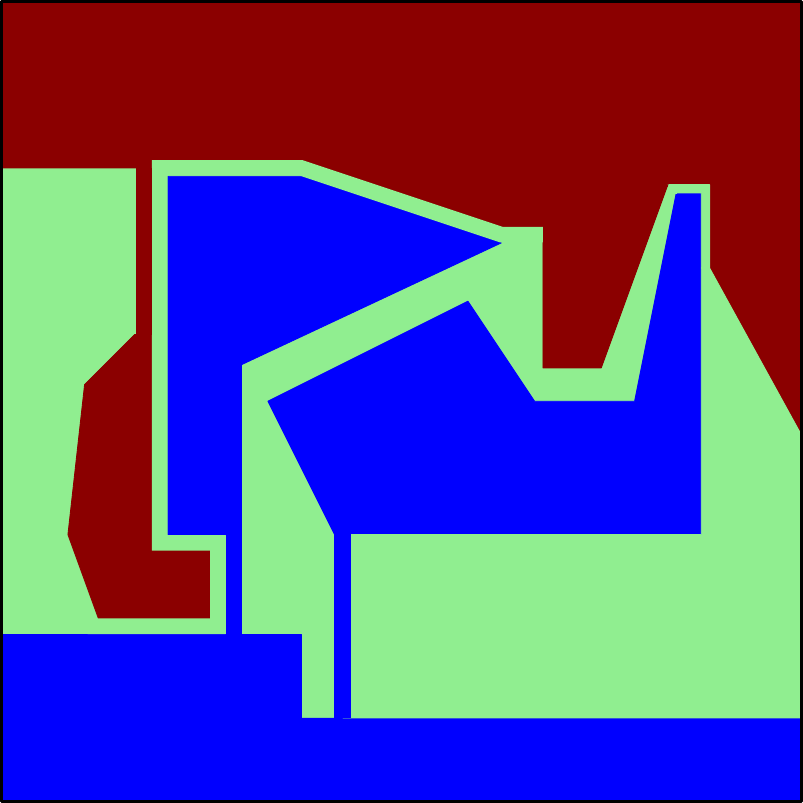}
    \end{minipage}
    \caption{An example of a weak embedding of a map. Left: multiple corridors connect disconnected regions. Right: the corridors are thickened to create three simple polygons.}
    \label{fig:thickening}
\end{figure}

For a district map $M$, the boundaries of the districts jointly form a straight-line embedding of some abstract graph $G$. By identifying edges on opposite sides of narrow corridors, we obtain a weak embedding of $G$. 
In a weak embedding, two or more corridors may overlap, and we maintain a linear order among all overlapping corridors. 

We use the machinery introduced by Akitaya et al.~\cite{akitaya2017recognizing}  (based on earlier work by Corsese et al.~\cite{cortese2009embedding} and Chang et al.~\cite{chang2014detecting}); see also~\cite{BlasiusFR21,FulekT22}. 
A weak embedding of $G$ is a piecewise linear map of $\varphi:G\to \mathbb{R}^2$.
The \emph{image graph} $H$ is a planar straight-line graph 
formed by the image $\varphi(G)$, where overlapping vertices (edges) of $G$ are mapped to the same vertex (edge). Given an image graph $H$ and a sufficiently small $\eps>0$, an \emph{$\eps$-thickening} of $H$ is a subset of $\mathbb{R}^2$ constructed as follows. For each $u\in V(H)$, create a disk $\mathcal{D}_\epsilon(u)$ centered at $u$ with radius $\eps$.
For each edge $uv\in E(H)$ the \emph{pipe} $N_\eps(uv)$ is the $\eps^2$-neighborhood of $uv$ minus $\mathcal{D}_\epsilon(u)\cup \mathcal{D}_\epsilon(v)$.
The $\eps$-thickening of $H$ is the union of all disks and pipes, which are pairwise disjoint for a sufficiently small $\eps>0$.

A \emph{weak representation} of $G$ comprises of a weak embedding $\varphi:G\to \mathbb{R}^2$ and a linear order of overlapping edges of $\varphi(G)$ along each edge of $H$.
Given a weak representation of $G$, an \emph{$\eps$-perturbation} is an embedding $\psi:G\to \mathbb{R}^2$ that maps $G$ into the thickening of $H$ so that the Fr\'echet distance between $\varphi$ and $\psi$ is at most $\eps$, and in each pipe $N_\eps(uv)$ the parallel edges between 
$\mathcal{D}_\epsilon(u)$ and $\mathcal{D}_\epsilon(v)$ are in the given linear order.  
It is known that if $G$ has $n$ vertices, then an $\eps$-perturbation $\psi(G)$ with $O(n)$ vertices can be computed in $O(n\log n)$ time~\cite{akitaya2017recognizing}.
The $\eps$-perturbation $\psi:G\to \mathbb{R}^2$ of a weak embedding $\varphi:G\to \mathbb{R}^2$ is \emph{area-preserving} if corresponding faces in $\psi(G)$ and $\varphi(G)$ have the same area.

%

We first state a lemma about three districts, and then generalize it to $k\geq 3$ districts.

\begin{lemma}
\label{lem:perturbation}
    Given a map $M(\mathcal{R}) = \{D_1, D_2, D_3\}$ in a weak representation with $O(n)$ vertices, and a sufficiently small $\eps'>0$, there exists an $\eps\in (0,\eps')$ with the following property: For any area-preserving 
    $\eps$-perturbation $D_1'$ of $D_1$ where $D_1'$ is a simple polygon with $O(n)$ vertices and $\mathcal{R}\setminus D_1'$ is connected, there exist area-preserving $\eps'$-perturbations of $D_2'$ and $D_3'$ of $D_2$ and $D_3$, resp., such that both $D_2'$ and $D_3'$ are polygons, and 
    $\{D'_1, D'_2, D'_3\}=M'(\mathcal{R})$ is a district map.
    \end{lemma}

\begin{proof}
    Let $H$ be the image graph obtained by the union of the boundary of the three districts.
    Let $\mathcal{T}_\eps$ and $\mathcal{T}_{\eps'}$ be the $\eps$- and $\eps'$-thickenings of $H$, respectively.
    Assume that $\eps'$ is small enough such that all the regions of $\mathcal{T}_{\eps'}$ are pairwise disjoint.
    Since all the boundaries are contained in $\mathcal{T}_\eps$, each component of $\mathcal{T}_{\eps'}\setminus \mathcal{T}_\eps$ is contained in a single district.
    Let $S$ denote the union of pipes $N_{\eps'}(uv)$ containing no edges of $D_1'$.
    We also assume that $\eps$
    is small enough so that
    \begin{equation}\label{eq:1}
        \area(\mathcal{T}_{\eps}) < \area(S\setminus \mathcal{T}_\eps).
    \end{equation}
    Let $D_2^*$ and $D_3^*$ be arbitrary 
    $\eps$-perturbations of $D_2$ and $D_3$  (not necessarily area-preserving), where $D_2^*$ and $D_3^*$ are simple polygons. We transform $D_2^*$ and $D_3^*$ into area-preserving $\eps'$-perturbations as follows.
    Since $D_2^*$ is an $\eps$-perturbation, then
    \begin{equation}\label{eq:2}
        |\area(D_2^*)-\area(D_2)|<\area(\mathcal{T}_{\eps}).
    \end{equation}
    Without loss of generality, assume $\area(D_2^*)>\area(D_2)$ and consequently $\area(D_3^*)<\area(D_3)$.
    For each pipe $N_{\eps'}(uv)\subseteq S$, subdivide all boundary edges in $N_{\eps'}(uv)$ at $\partial\mathcal{D}_{\eps}(u)$, $\partial\mathcal{D}_{\eps}(v)$, $\partial\mathcal{D}_{\eps'}(u)$ and $\partial\mathcal{D}_{\eps'}(v)$.
    We can continuously move all vertices at $\partial\mathcal{D}_{\eps'}(u)$ and $\partial\mathcal{D}_{\eps'}(v)$ so that the districts remain simple polygons,
    and so that in the limit the entire pipe $N_{\eps'}(uv)$ belongs to $D_3^*$ at the end of the motion; see Figure~\ref{fig:perturbation}. If we apply such continuous deformations simultaneously in all pipes in $S$, then at the end of the motion we have $\area(D_3^*)>\area(D_3)$ and $\area(D_2^*)<\area(D_2)$ because of (\ref{eq:1}) and (\ref{eq:2}).
    As the areas of regions change continuously in a continuous deformation, by the intermediate value theorem, there exist area-preserving $\eps'$-perturbations $D_2'$ and $D_3'$, as required.
\end{proof}

\begin{figure}[htb]
	\centering
    \begin{minipage}{.8\textwidth}
          \centering
          \includegraphics[width = 0.9 \textwidth]{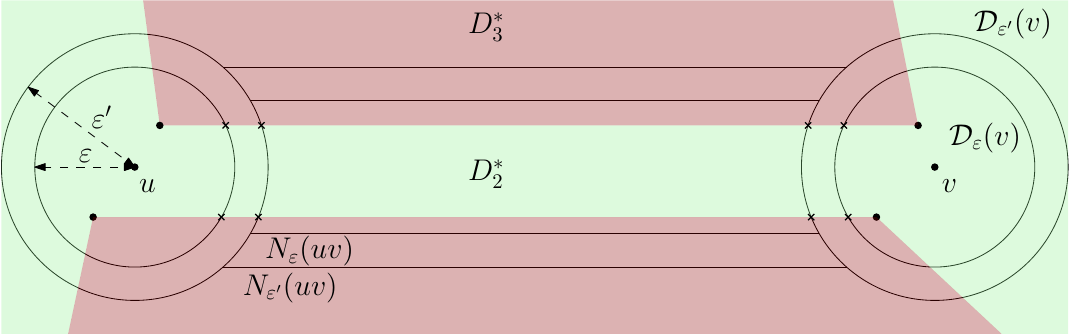}
    \end{minipage}
    \begin{minipage}{.8\textwidth}
          \centering
          \includegraphics[width = 0.9 \textwidth]{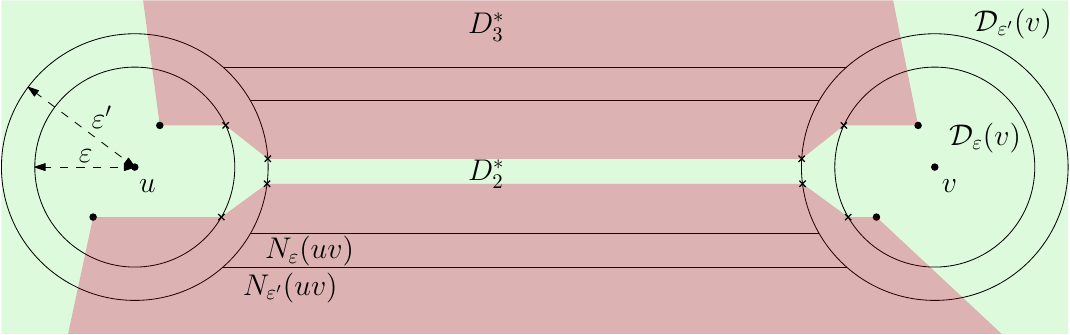}
    \end{minipage}
    \caption{We subdivide all boundary edges in a pipe $N_{\eps'}(uv)\subseteq S$ at $\partial\mathcal{D}_{\eps}(u)$, $\partial\mathcal{D}_{\eps}(v)$, $\partial\mathcal{D}_{\eps'}(u)$ and $\partial\mathcal{D}_{\eps'}(v)$;
    and then move the vertices at $\partial\mathcal{D}_{\eps'}(u)$ and $\partial\mathcal{D}_{\eps'}(v)$ to increase the area of $D_3^*$.}
    \label{fig:perturbation}
\end{figure}

\begin{proposition}
\label{cor:perturbation}
    Given two area-compatible $k$-district maps and a sequence of area-preserving ReCom moves where districts in intermediate maps are weak polygons with $O(n)$ vertices, we can compute a sequence of area-preserving ReCom moves of the same length where the districts in intermediate maps are all  polygons with $O(n)$ vertices.    
\end{proposition}

\begin{proof}
    First note that we can generalize Lemma~\ref{lem:perturbation} to $k$-district maps by grouping $k-2$ districts into one single weak polygon.
    The application of Lemma~\ref{lem:perturbation} will give a compatible perturbation of the union of the $k-2$ districts into a polygon. We can now recurse by redefining $\mathcal{R}$ to be such a polygon.
    Now, compute the image graph for each of the intermediate maps.
    Each image graph imposes a constraint on $\eps'$ and $\eps$ given by Lemma~\ref{lem:perturbation}.
    For two subsequent maps $M_1$ and $M_2$, let $\eps'_1 > \eps_1$, and $\eps'_2 > \eps_2$ be the perturbation amounts in Lemma~\ref{lem:perturbation}. 
    Without loss of generality, $D_1$ is the union of districts not changed in the ReCom that transforms $M_1$ into $M_2$.
    We then make $\eps_1'=\eps_2$, which imposes an additional upper bound on $\eps_1'$. 
    We can then compute all constraints on the perturbation amounts from the last to the first map. 
    We then set the perturbation amounts for the first intermediate map and propagate it to the others.
    Now Lemma~\ref{lem:perturbation} implies the claim.
\end{proof}

\subparagraph{Weak Representation for ReCom Sequences.} 
We construct a ReComb sequence in two passes: The first pass operates on a generic  $\eps$-perturbation, where the \emph{area} of each district is given by the weak representation (hence the corridors have zero area). The second pass then expands the weak representations into an $\eps$-perturbations, using Proposition~\ref{cor:perturbation}, where each district is a simple polygon with the desired area. Note that the number of moves is determined in the first pass. 

We define the \emph{compressed complexity} of a district map as the number of vertices in the image graph $H$ of the weak representation (that is, repeated vertices are counted only once).
The number of ReCom moves produced by our algorithm in Sections~\ref{sec:three}--\ref{sec:general} depends on the compressed complexity. 
Using $\eps$-perturbations would increase the complexity of maps.
For this reason, it is also useful in our analysis to convert an $\eps$-perturbations to a weak embedding which we do by applying the inverse of the operations described here.
Throughout this paper we use set operations on weak polygons such as $D_1\cup D_2$ where $D_1$ and $D_2$ are weak polygons.
Let $D_1'$ and $D_2'$ be the polygons obtained by the $\eps$-perturbation defined in Proposition~\ref{cor:perturbation}.
We define $D_1\cup D_2$ to be the weak polygon obtained from $D_1'\cup D_2'$.

    

\section{Reconfiguration for Three Districts}
\label{sec:three}

In this section, we consider maps with three districts in a unit square domain $\mathcal{R}=[0,1]^2$. We show that any 3-district map $M(\mathcal{R})=\{D_1, D_2,D_3\}$ can be transformed by finite sequence of ReCom moves into an area-compatible \emph{canonical map} in which the districts are axis-aligned rectangles, $Q_1$, $Q_2$ and $Q_3$, of unit width such that $\area(Q_i)=\area(D_i)$ for $i=1,2,3$. 

Before describing the algorithm for producing a chain of ReCom moves connecting two maps, we provide a high-level overview of where the main obstacles lie. We can gain some intuition by first considering the case when connectivity is not enforced, i.e., where districts need not be connected. Also assume for simplicity that the areas of the districts are all equal to $\frac{1}{3}$. In this setting, three ReCom moves suffice. Indeed, recombine and split $D_1$ and $D_2$ so that the new $D_1$ does not intersect $Q_3$ (this is always possible since the area of $D_3 \cup Q_3$ is at most $\frac{2}{3}$). We can now recombine $D_2$ and $D_3$ so that the new $D_3$ equals $Q_3$, and the third ReCom move can complete the transformation. Let us now return to the case where the district must remain connected. The first main insight shown in this section is that the three-step process above can still be applied, just by adding narrow corridors to maintain connectivity at each stage. However, to reach canonical form, we must remove these corridors, which requires care. The main difficulty here comes from the fact that corridors usually run alongside each other. The removal of any such corridor would require at least two ReCom moves. The second insight is a method to remove such corridors suing a number of ReCom moves bounded by a function of the complexity of the initial map.

\subsection{Overview of the Algorithm} 

Our algorithm for transforming a map into the canonical map consists of three stages, each containing multiple ReCom moves:
\begin{itemize}
    \item \emph{Preprocessing} (Section \ref{ssec:preprocess}). In this stage, we ensure that our three districts are ordered top to bottom in a well-defined way, and the middle district has the largest area. Moves needed: $O(1)$.
    \item \emph{Gravity moves} (Section \ref{ssec:gravity}). We perform three ReCom moves to place the districts into their final positions, with the possible exception of corridors. Moves needed: $3$.
    \item \emph{Exchange sequences} (Section \ref{ssec:exchange}). Corridors maintaining connectivity are carefully removed, using a tree representation to determine a move that simultaneously removes a constant fraction of corridors. Moves needed: $O(\log n)$. 
\end{itemize}

\subsection{Preprocessing: Ordering Property} \label{ssec:preprocess}
First we transform the three given districts into simple polygons if necessary. While there is a district $D_i$ that is a polygon with holes, there is an adjacent district $D_j$ contained within a hole. Recombine $D_i$ and $D_j$ to create a single-edge corridor between $D_j$ and the outer boundary of $D_i$. Next, we create corridors, if necessary, such that each district touches both the left and right sides of $\mathcal{R}$. While there is a district that is not adjacent to the left (resp., right) side $s$ of the $\mathcal{R}$, let $D_i$ be such a district closest to $s$ and let $D_j$ be an adjacent district that already touches $s$; then we recombine $D_i$ and $D_j$ and append to $D_i$ a shortest path to $s$ along the boundary of $D_j$. Thus, both districts remain simply connected. As all corridors run along existing vertices of the three districts, the complexity of the map does not increase. This stage takes $O(1)$ ReCom moves.

After preprocessing, the intersection of each district with the left (resp., right) side of the square domain is connected; and the order of these intersections is the same on both sides, or else two districts would cross. Therefore, the districts can be ordered from top to bottom.
%
%
We also need to establish the property that the middle district has the largest area. This can be done trivially with a single ReCom move between the middle district and the largest district of the three. We call these properties combined the \emph{ordering property}:

\begin{definition}
A three district map $\map(\mathcal{R}) = \{D_1,\ldots,D_k\}$ satisfies the \textbf{ordering property} if the intersection of each district with the left (resp., right) side of the square domain is connected, and the middle district, as defined by the resulting order from top to bottom, has area greater than or equal to each other district. 
\end{definition}


%

    

We assume that the districts are simple polygons in the unit square with a total of $n$ vertices and describe the details of the recombination moves as we use them in the algorithm. To reconfigure the districts into their canonical positions, apart from possible corridors, we perform three gravity moves.
The gravity move is described next.

\subsection{Gravity Move} \label{ssec:gravity}
Assume that $\map$ is a 3-district map satisfying the ordering property, with districts labeled $D_1$ (red), $D_2$ (green), and $D_3$ (blue) from top to bottom.
We describe the move $\gravity(D_1,D_2)$, which recombines the red and green districts; refer to Figures~\ref{fig:gravity_ex1}--\ref{fig:gravity_ex}.
Let $P=D_1\cup D_2$, which is a weakly simple polygon by the ordering property. 
By continuity, there exists a horizontal line $\ell$ (that we call the \emph{waterline}) that
partitions the plane into upper and lower halfplanes $\ell^+$ and $\ell^-$, resp.,
such that $\area(P\cap \ell^+)=\area(D_1)$ and $\area(P\cap \ell^-)=\area(D_2)$.
We shall define new districts $D_1'$ and $D_2'$, resp., that contain $P\cap \ell^+$ and $P\cap \ell^-$.

Note, however, that $P\cap \ell^+$ and $P\cap \ell^-$ may be disconnected.
We then reconnect disjoint components of each district by corridors along the boundary of $P$; see Fig.~\ref{fig:gravity_ex}.
Note that, by the ordering property, there is a path $\pi$ on the boundary of $D_3$ (blue) between the left and right side of the domain $\mathcal{R}$. 
If there are two or more components of $P\cap \ell^+$, they are separated by blue and, therefore they all touch the path $\pi$. Therefore $(P\cap \ell^+)\cup \pi$ is a connected region. Similarly, $(P\cap \ell^-)\cup \pi$ is also connected.

We define a \emph{red graph} as follows: the vertices are the connected components of $P\cap \ell^+$ and edges are minimal arcs along $\pi\cap \ell^-$ that connect two distinct components of $P\cap \ell^+$. 
Since $(P\cap \ell^+)\cup \pi$ is connected, then the red graph is connected.
Consider an arbitrary spanning tree of the red graph, and add its edges (as corridors) to the red district along the boundary of $P$. This completes the definition of $D_1'$. 

\begin{figure}[h]
	\centering
    \begin{minipage}{.32\textwidth}
          \centering
          \includegraphics[width = 0.8 \textwidth]{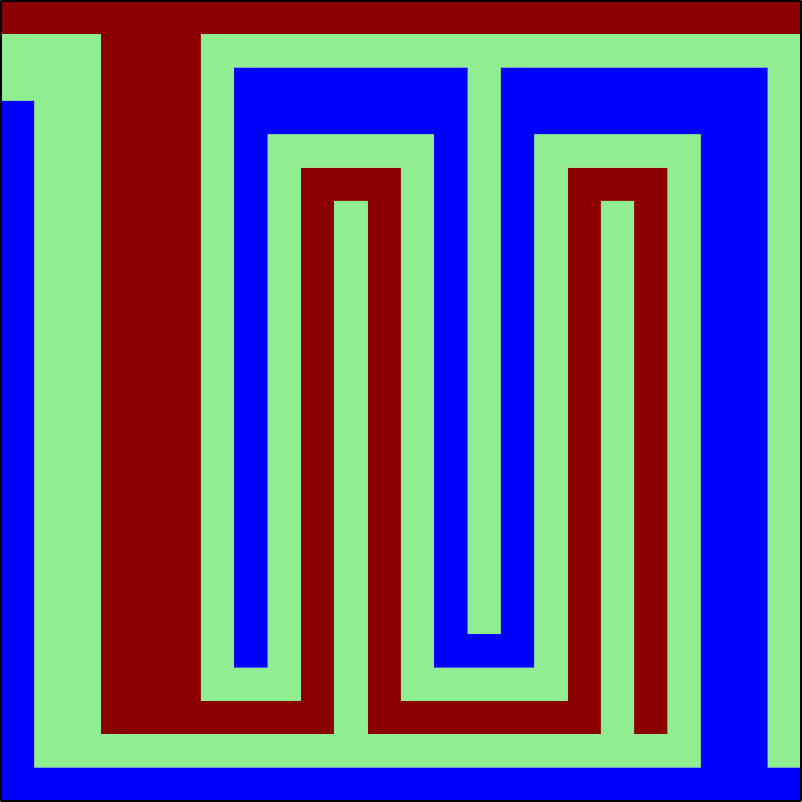}
    \end{minipage}
    \begin{minipage}{.32\textwidth}
          \centering
          \includegraphics[width = 0.8 \textwidth]{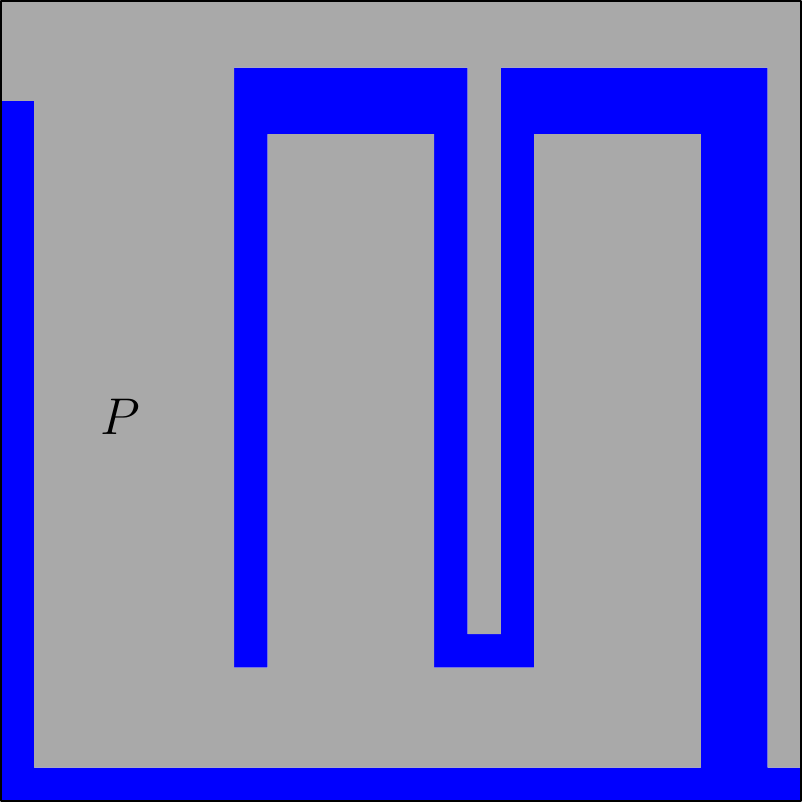}
    \end{minipage}
    \begin{minipage}{.32\textwidth}
          \centering
          \includegraphics[width = 0.8 \textwidth]{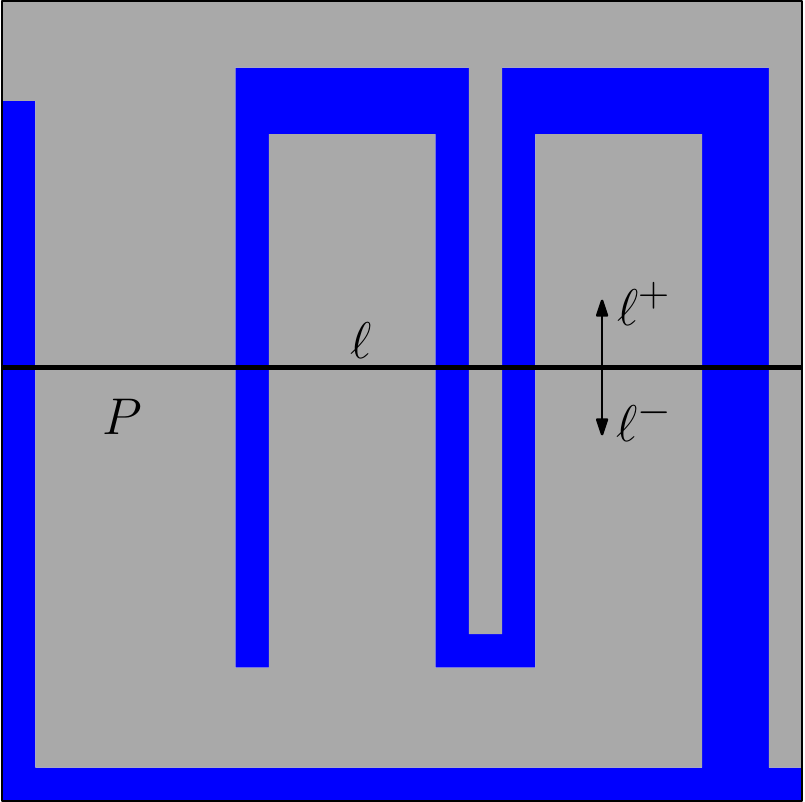}
    \end{minipage}
    \caption{The setup for a gravity move between the red district $D_1$ and green district $D_2$. Left: a district map satisfying the ordering property. Middle: the union $P=D_1\cup D_2$ is shown in gray. Right: the horizontal line $\ell$ equipartitions the gray polygon $P$. }
    \label{fig:gravity_ex1}
\end{figure}

\begin{figure}[h]
	\centering
    \begin{minipage}{.32\textwidth}
          \centering
          \includegraphics[width = 0.8 \textwidth]{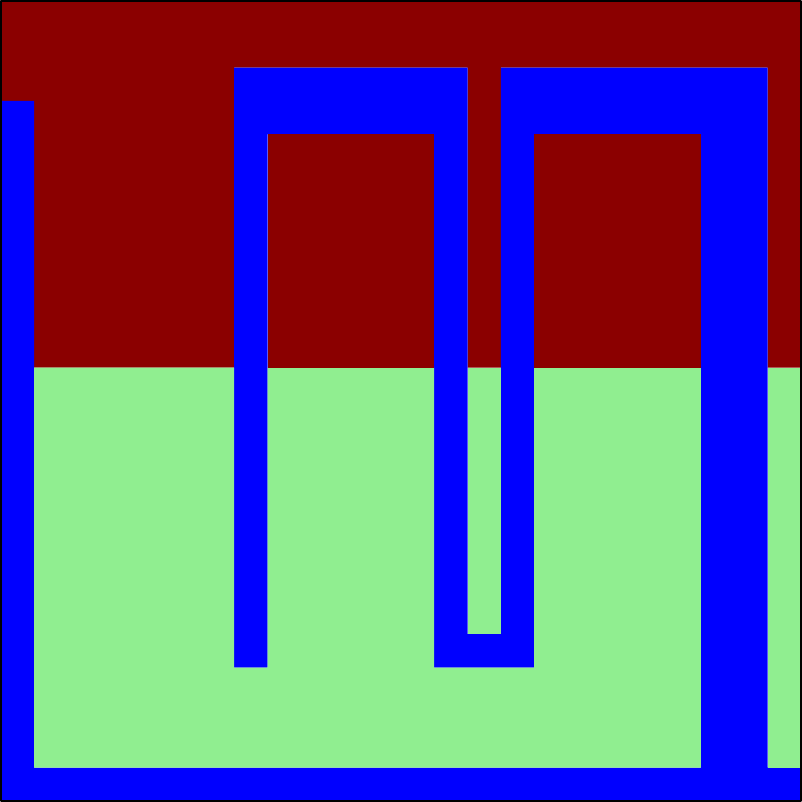}
    \end{minipage}
    \begin{minipage}{.32\textwidth}
          \centering
          \includegraphics[width = 0.8 \textwidth]{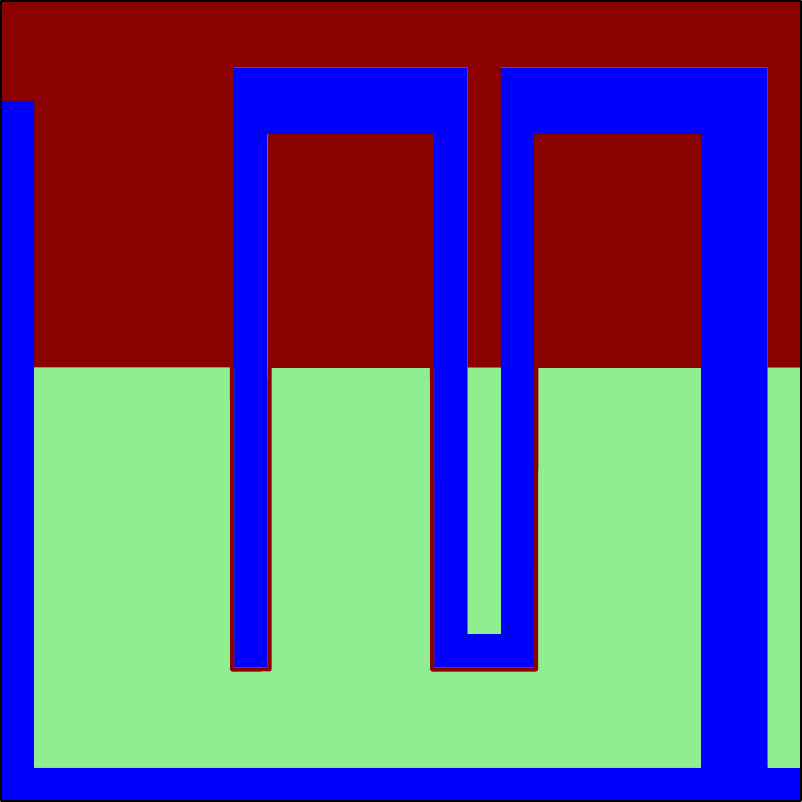}
    \end{minipage}
    \begin{minipage}{.32\textwidth}
          \centering
          \includegraphics[width = 0.8 \textwidth]{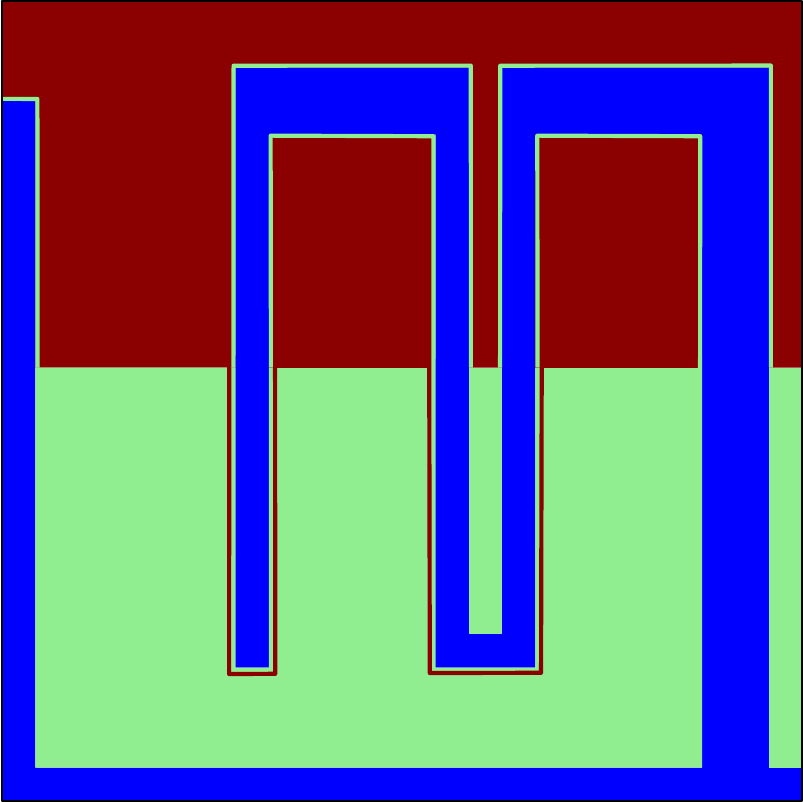}
    \end{minipage}
    \caption{Constructing the result of a gravity move between red and green on the map in Figure \ref{fig:gravity_ex1}. Left: the red region $P\cap \ell^+$ and the green region $P\cap \ell^-$ are each disconnected. Middle: red corridors create a connected red district $D_1'$. Right: green corridors create a connected green district $D_2'$ and restore the ordering property.}
    \label{fig:gravity_ex}
\end{figure}


Since the blue district is simply connected, each component of $P\cap\ell^-$ also intersects $\ell$ and therefore is adjacent to the red district $D_1'$.
Intuitively, we add corridors along $\pi$ ``coating'' the blue district with green and thus restoring the ordering property. 
Note that $\pi$ may pass along the boundary of $D_1'$, including all red corridors,  and the boundaries of the components of $P\cap \ell^-$. 
Formally, we add green corridors at the intersection of $\pi$ and $\partial D_1'$, if such a corridor is parallel to a red corridor, it runs between the blue district and the red corridor. That defines $D_2'$ and concludes the description of the gravity move.

\begin{restatable}{lemma}{lemgravityondisk} \label{lem:gravity_on_disk}
Assume $D_1$ and $D_2$ are the top two districts on a map satisfying the ordering property. Then $\gravity(D_1,D_2)$ is an area-preserving ReCom move that maintains the ordering property.
\end{restatable}

\begin{proof}
By construction, $D_1'$ is simply connected since it is homotopic to a spanning tree of the red graph. Furthermore $D_1'$ contains the upper-left and upper-right corners of the unit square, so it touches both sides of the square. As noted above, $(P\cap \ell^-)\cup \pi$ is a connected region, each component of $P\cap \ell^-$ is adjacent to $\pi$. Therefore, $D_2'$ is connected and touches both left and right sides of the domain. Thus, the ordering property is maintained.
\end{proof}

Since each waterline intersects an edge of a district at most once we have:

\begin{restatable}{lemma}{lemgravitycomplexity}\label{lem:gravitycomplexity}
Assume $D_1$, $D_2$ and $D_3$ each have at most $m$ vertices. Then \gravity($D_1$,$D_2$) produces districts $D'_1$ and $D'_2$, each with at most $O(m)$ vertices.
\end{restatable}

\begin{proof}
We show that all steps of the \gravity($D_1$,$D_2$) move create $O(m)$ vertices. 
The waterline $\ell$ intersects $\pi$ in at most $m$ points, which become vertices of both $D_1'$ and $D_2'$. 
Then red corridors are added along the edges of an arbitrary spanning tree of the red graph. At most one red corridor passes through any vertex of $\pi$, which has at most $m$ vertices. As the red corridors have at most two vertices at each vertex of $\pi$, this step created $O(m)$ new vertices for $D'_1$.
When green corridors are added, at most $2$ corridors pass through any vertex of $\pi$, which has at most $m$ vertices. As before, the green corridors have $O(1)$ vertices at each vertex of $\pi$, thus this step created $O(m)$ new vertices for $D'_2$.
\end{proof}

The move $\gravity(D_3,D_2)$ is defined analogously: a reflection in a horizontal line that reverses the order of the three districts, such that $D_1'=D_3$ becomes the top district and $D_2'=D_2$ is the middle district, then apply $\gravity(D'_1,D'_2)$, followed by another reflection.




\begin{lemma}
    \label{lem:gravity}
    Let $M$ be a map  satisfying the ordering property, with districts $D_1$, $D_2$ and $D_3$ from top to bottom.
    Then $\gravity(D_1,D_2)$ returns a map that satisfies the ordering property and $D_1'$ is disjoint from $Q_3$ with the possible exception of corridors,  where $Q_3$ is the axis-aligned rectangle of the blue district in the canonical map.
\end{lemma}
\begin{proof}
It suffices to show that the horizontal line $\ell$ is above $Q_3$.
Lemma~\ref{lem:gravity_on_disk} yields the rest. By definition, the area below $\ell$ is at least $\mathrm{area}(D_2) \geq \mathrm{area}(Q_3)$, since $D_2$ has the maximum area of the three districts. Thus, the line $\ell$ is above $Q_3$.
\end{proof}

\begin{figure}[htb]
	\centering
    \begin{minipage}{.45\textwidth}
          \centering
          \includegraphics[width = 0.8 \textwidth]{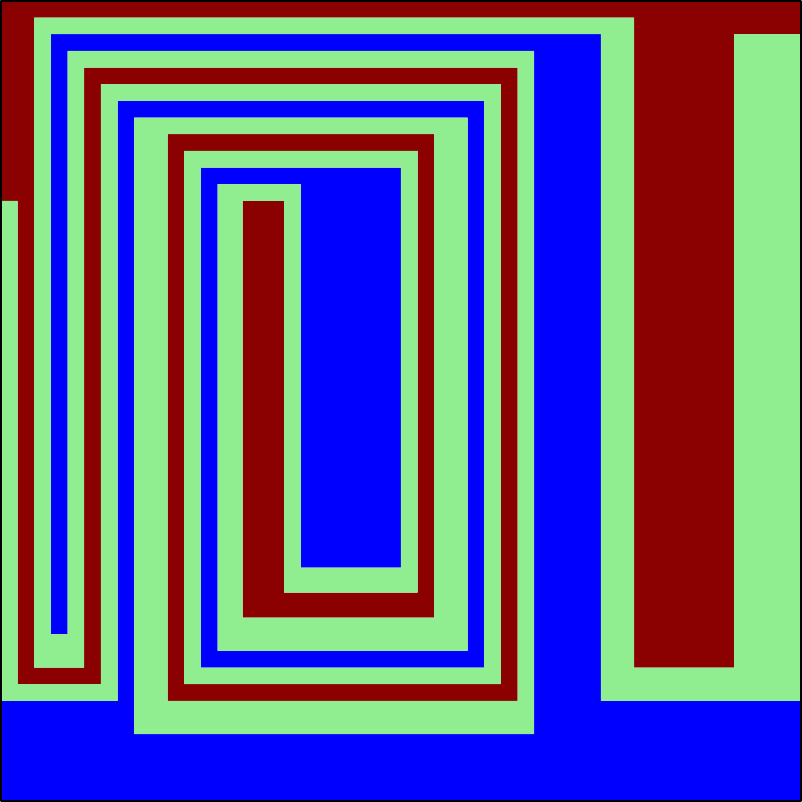}
    \end{minipage}
        \vspace{2mm}
    \begin{minipage}{.45\textwidth}
          \centering
          \includegraphics[width = 0.8 \textwidth]{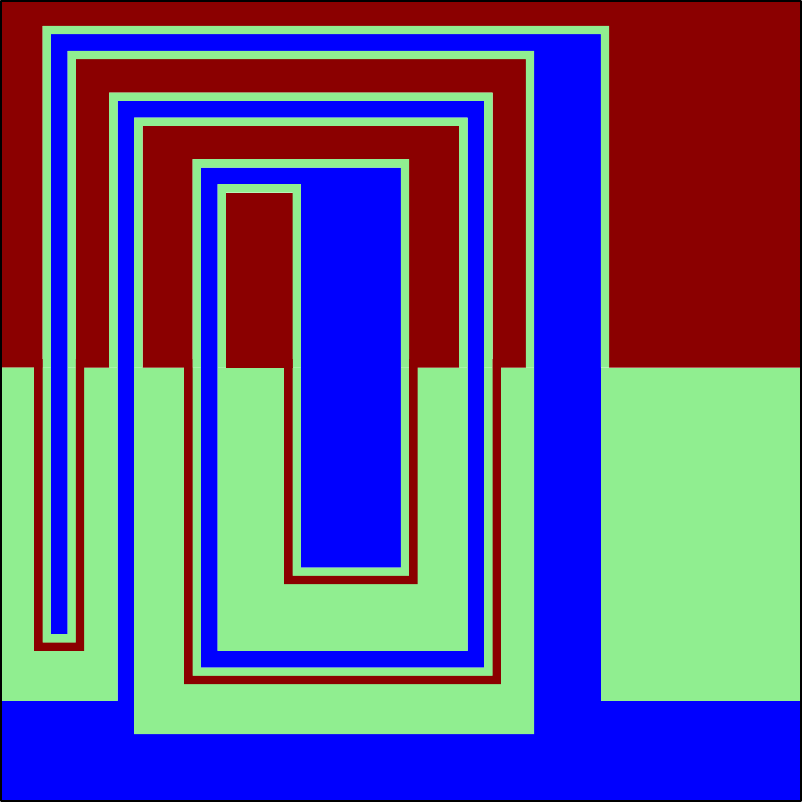}
    \end{minipage}
    \begin{minipage}{.45\textwidth}
          \centering
          \includegraphics[width = 0.8 \textwidth]{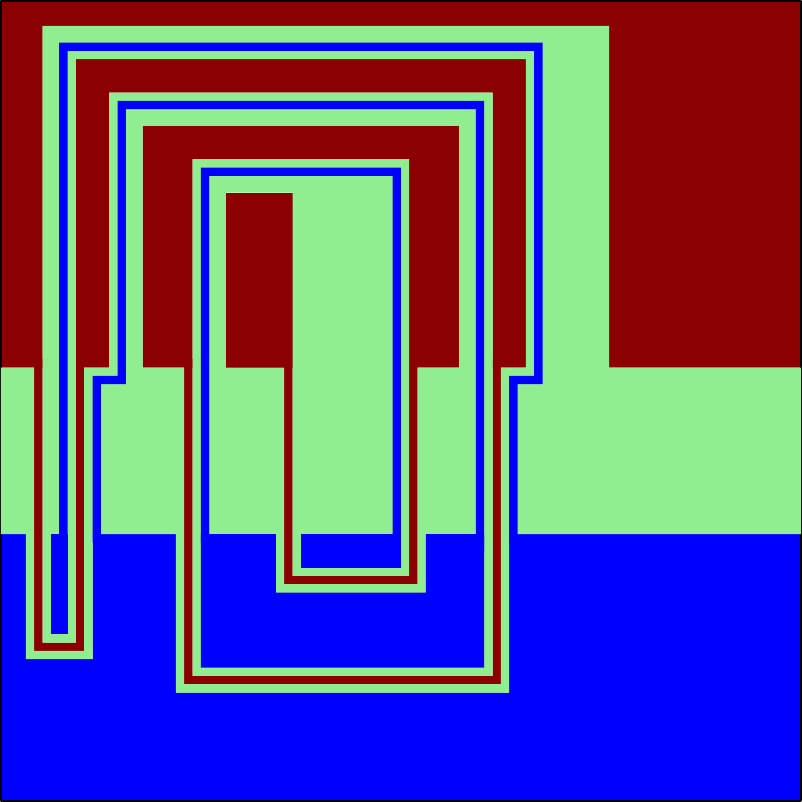}
    \end{minipage}
            \vspace{2mm}
    \begin{minipage}{.45\textwidth}
          \centering
          \includegraphics[width = 0.8 \textwidth]{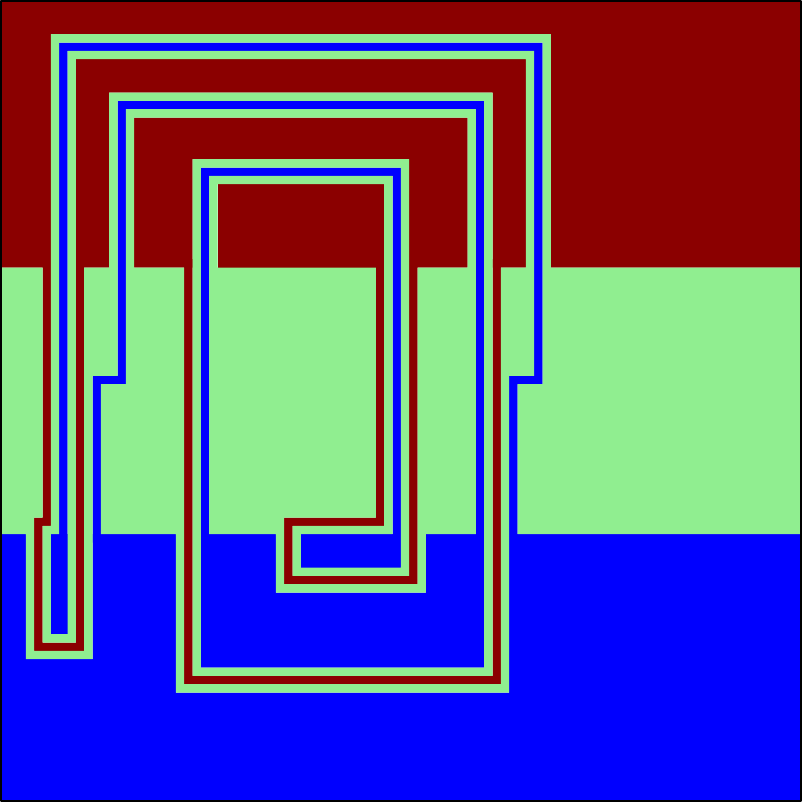}
    \end{minipage}
    \caption{An example of a sequence of three gravity moves. Top-left: a starting configuration. Top-right: the result of $\gravity(D_1,D_2)$. Bottom-left: the result of $\gravity(D_3,D_2)$. Bottom-right: the result of the third gravity move $\gravity(D_1,D_2)$.}
    \label{fig:after2gravity}
\end{figure}

After three gravity moves each district has positive area only in the regions in their corresponding districts in the canonical configuration (see Figure~\ref{fig:after2gravity} for an example).

\begin{restatable}{lemma}{lemthreegravity}\label{lem:3gravity}
Let $M$ be a map  satisfying the ordering property, with districts $D_1$, $D_2$ and $D_3$ from top to bottom.
Then the sequence of three moves $\gravity(D_1, D_2)$, $\gravity(D_2, D_3)$, and $\gravity(D_1, D_2)$ return a map $M'$ where $D_1$, $D_2$, and $D_3$ are each contained in their canonical rectangles, with the possible exception of some corridors.
\end{restatable}

\begin{proof}
Let $Q_1$, $Q_2$ and $Q_3$ be the rectangles of the three districts in the canonical map (that is, axis-aligned rectangles of unit width with $\area(Q_i)=\area(D_i)$ for $i=1,2,3$).

From Lemma~\ref{lem:gravity} we know that after the first gravity move, $D_1$ is disjoint from $Q_3$ with the possible exception of corridors. Thus, the second gravity move between $D_2$ and $D_3$ moves $D_3$ into $Q_3$ with the possible exception of corridors. Finally, since $D_1 \cup D_2$ is entirely contained in $Q_1 \cup Q_2$ with the possible exception of corridors, the third gravity move between $D_1$ and $D_2$ moves each of them into their final position with the possible exception of corridors.
\end{proof}

\subsection{Tree Representation of a Region.}
After prepocessing and the three gravity moves in Lemma~\ref{lem:3gravity}, we want to eliminate corridors. We encode the topology of the region $P = D_1 \cup D_2$ in a graph that we use for the $\exchange$ sequence, described below.

We define the \emph{corridor graph} $T(R)$ of a weakly simple polygon $R \subset \mathcal{R}$. A weakly simple polygon has a unique decomposition into pairwise disjoint  simple polygons and corridors (polygonal paths). The nodes of $T(R)$ are simple polygons in the decomposition of $R$, and the edges represent corridors between two polygons in $R$. Denote the set of edges by $E(T(R))$. At each node, the rotation of the incident edges represents the counterclockwise order of corridors along the corresponding polygon in $R$. The weight of each node is the area of the corresponding polygon. As corridors have zero thickness, the total weight of the nodes is $W=\area(R)$.

In particular, we want to consider the corridor graph of $P = D_1 \cup D_2$. Assume that $\map$ is a 3-district map returned by the three gravity moves in Lemma~\ref{lem:3gravity}. By the ordering property, we know that the intersection of $D_1$ and $D_2$ is a simple path - either from one side of the square to another or, if $D_1$ is contained in $D_2$, then it is a closed curve. Thus, $P$ is a weakly simple polygon. Let $Q_{12}$ be the union of the two axis-aligned rectangles that contain $D_1$ and $D_2$ in the canonical configuration. Then, the nodes of $T(P)$ are simple polygons in $P \cap Q_{12}$ (regions bounded by corridors of $D_3$) and the edges are corridors in $R \setminus Q_{12}$ that connect two such polygons (corridors of $D_1$ and $D_2$ running through $Q_3$).
Note, however, that a corridor in $P$ may be the union of three parallel corridors in $D_2$, $D_1$, and $D_2$, resp.; see Fig.~\ref{fig:corridor_graph_def}.
Since $P$ is a weakly simple polygon, $T(P)$ is a tree; see Fig.~\ref{fig:corridor_graph_def}.
Note that the number of vertices in $T(P)$ is bounded above by the compressed complexity of the map and that many different maps can have the same corridor graph.

\begin{figure}[ht!]
	\centering
    \begin{minipage}{.48\textwidth}
          \centering
          \includegraphics[width = 0.8 \textwidth]{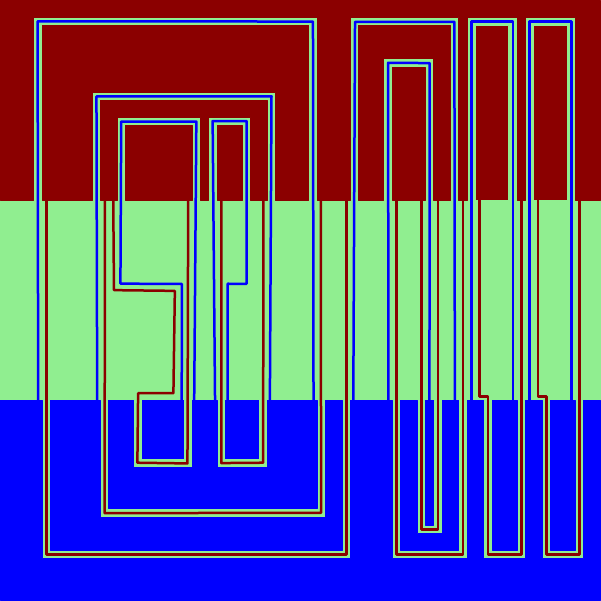}
          
          (a)
    \end{minipage}
    \centering
    \begin{minipage}{.48\textwidth}
          \centering
          \includegraphics[width = 0.8 \textwidth]{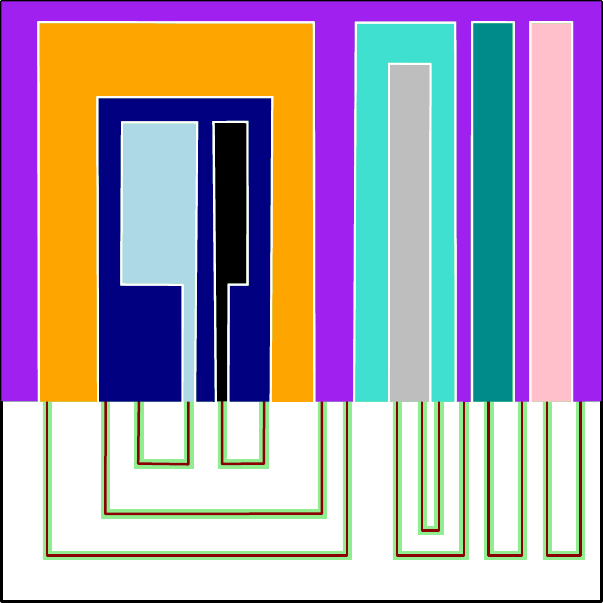}
          
          (b)
    \end{minipage}
    \centering
    \includegraphics[width = 0.5 \textwidth]{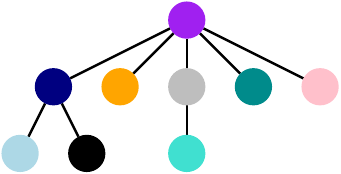}
    
    (c)
    \caption{(a) A map $M$ after 3 \gravity\ moves. (b) The nodes of the corridor graph $T(P)$ correspond to connected components of $P \cap Q_{12}$, indicated by distinct colors.
    (c) The corridor graph $T(P)$ encodes the topology of $P$.
    }
    \label{fig:corridor_graph_def}
\end{figure}


We use the corridor graph $T(P)$ to eliminate corridors. Consider what happens if the tree has a leaf that is entirely part of the green district (see Fig.~\ref{fig:afterexchange}). This means that by doing a gravity move between green and blue we can eliminate the green and blue corridors adjacent to this leaf, removing the leaf from the tree altogether. Our goal is therefore to create a part of the tree which is entirely green.

The \emph{centroid} of a vertex-weighted tree of total weight $W$ is a vertex whose removal partitions the tree into subtrees of weight at most $\frac{W}{2}$ each. Jordan~\cite{Jordan69} proved that every tree (with unit weights) has a centroid; this was perhaps the oldest separator theorem~\cite{LT79,MillerTTV97}. The result extends to weighted trees: a greedy algorithm finds the centroid in linear time.

Let $c$ be a centroid of $T(P)$, and assume that $T(P)$ is rooted at $c$.
Let $t$ be the node of $T(P)$ that corresponds to the region of $P$ containing the top edge of the square domain.
A  subtree of $T(P)$ is \emph{contiguous} if it consists of the centroid $c$, some children of $c$ that are consecutive in the rotation order of $c$, and all their descendants in $T(P)$. 

\begin{restatable}{lemma}{threesplit} \label{lem:3split}
There exists a contiguous subtree $T^*$ of $T(P)$ such that:
{\rm (i)} $T^*$ contains at least $\frac{1}{3}$ of the vertices of $T(P)$, and
{\rm (ii)} the weight of $T^*$ is at most $\frac{W}{2} + w(c)$, where $w(c)$ is the weight of the centroid $c$.
\end{restatable}

\begin{proof}
By the definition of the centroid $c$, the removal of $c$ produces a forest $T(P)-c$, where the weight of each component (tree) is at most $\frac{W}{2}$.
Partition these $\deg(c)$ trees into up to three forests of consecutive subtrees such that each forest has weight at most $\frac{W}{2}$ as follows. 
Begin with a partition into $\deg(c)$ forests, each containing a single tree, and maintain their cyclic order around $c$. While there are two consecutive forests whose combined weight is at most $\frac{W}{2}$, merge them into a single forest.
The while loop terminates with three or fewer forests: Indeed, for four or more forests, the combined weight of at least one of the consecutive pairs would be at most $\frac{W}{2}$ by the pigeonhole principle.
Since we partition $T(P)-c$ into three forests, one of them contains at least $\frac{1}{3}$ of the vertices $T(P)-c$. Adding $c$ to this forest, we obtain a contiguous subtree of $T(P)$ containing at least $\frac{1}{3}$ of the vertices of $T(P)$.
\end{proof}


\begin{figure}[ht!]
    \centering
    \begin{minipage}{.44\textwidth}
          \centering
          \includegraphics[width = 0.85 \textwidth]{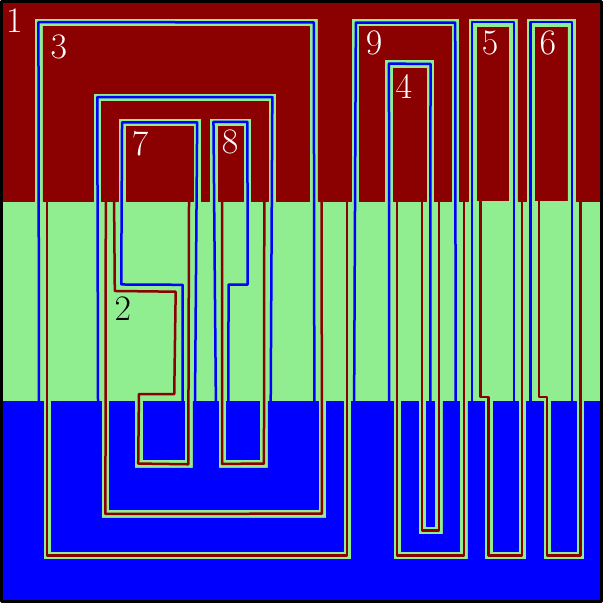}
    \end{minipage}
    \centering
    \begin{minipage}{.54\textwidth}
          \centering
          \includegraphics[width = 0.85 \textwidth]{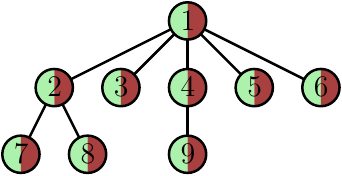}
    \end{minipage}
    \caption{The setup for an exchange sequence. Left: a map returned by a sequence of three gravity moves. Right: the corridor graph $T(P)$ that encodes the topology of $P=D_1\cup D_2$. Every node corresponds to a polygon in $P$. In this example, each polygon intersects both districts, which is indicated by a split node-coloring.}
    \label{fig:beforeexchange}
\end{figure}

\subsection{Exchange Sequence}\label{ssec:exchange}
We now describe the \emph{exchange sequence}, a sequence of three ReCom moves, which eliminates a fraction of the corridors and reduces the (compressed) complexity of the map.
Assume we are given a 3-district map $\map$ satisfying the ordering property.
As before, label its districts red, green, and blue from top to bottom.
We further require that there exist two horizontal lines $\ell_1$ and $\ell_2$ such that red has positive area only above $\ell_2$, blue has positive area only below $\ell_1$, and green has positive area only between $\ell_1$ and $\ell_2$ (cf.~Lemma~\ref{lem:3gravity}). 
See Fig.~\ref{fig:beforeexchange} for an example.

\begin{figure}[ht!]
    \centering
    \begin{minipage}{.44\textwidth}
          \centering
          \includegraphics[width = 0.85 \textwidth]{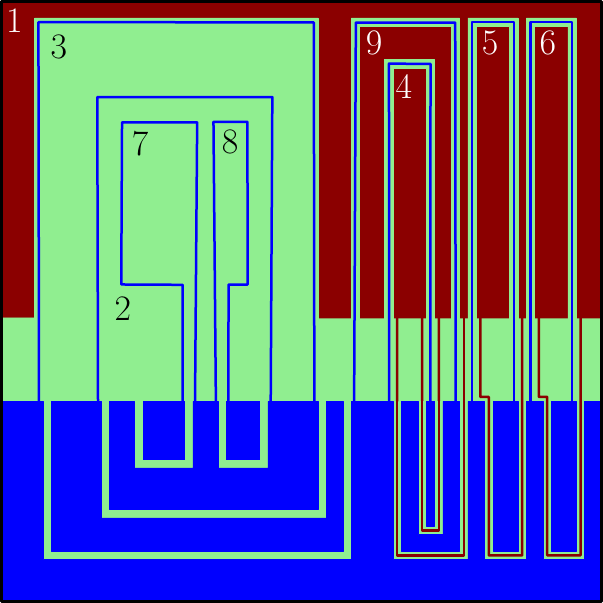}
    \end{minipage}
    \centering
    \begin{minipage}{.54\textwidth}
          \centering
          \includegraphics[width = 0.85 \textwidth]{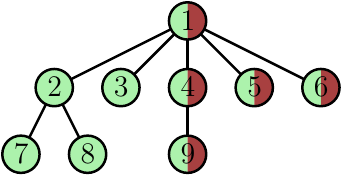}
    \end{minipage}
    \caption{The setup for an exchange sequence. Left: a map returned by a sequence of three gravity moves. Right: the corridor graph $T(P)$ that encodes the topology of $P=D_1\cup D_2$. Every node corresponds to a polygon in $P$. In this example, each polygon intersects both districts, which is indicated by a split node-coloring.}
    \label{fig:afterexchange}
\end{figure}


Let $c$ be a centroid of $T(P)$, where $P = D_1 \cup D_2$ and let $T^*$ be a contiguous subtree of $T(P)$ rooted at the centroid, as in Lemma~\ref{lem:3split}. The exchange sequence consists of the following three ReCom moves:
\begin{enumerate}
    \item ReCom green and red: Let $Q$ denote the regions of $T^*$ except for the region corresponding to node $c$. First make $Q$ green. Then partition the remaining region $P \setminus Q$ with a gravity-like move as follows.
    Apply a $\gravity$ move w.r.t.\ $P\setminus Q$ to subdivide it into two weak polygons of areas $\area(D_1)$ for red and $\area(D_2) - \area(Q)$ for green; see Fig.~\ref{fig:afterexchange}. After this ReCom move, $D_1$ is weakly simple and $D_2$ is a weak polygon (in which $D_1$ is a hole if $Q$ is a weak polygon with a hole). 
    
    
    \item ReCom green and blue removing unnecessary green and blue corridors simultaneously as follows.
    Remove any green and blue monochromatic corridors corresponding to all edges of $T^*$. Note that this merges some nodes of $T(D_3)$ (see Fig.~\ref{fig:afterexchangegravity}), and creates cycles in $T(D_3)$. While there is a cycle in $T(D_3)$ remove a blue corridor in an edge of $T(D_3)$ in a cycle. As this process modifies only green and red, it requires a single ReCom move. After this ReCom move, $D_3$ is a weakly simple polygon and $D_2$ is a weakly simple or weak polygon. 
    
    \item ReCom green and red with a $\gravity$ move, restoring the ordering property.
\end{enumerate}




\begin{figure}[ht!]
    \centering
    \begin{minipage}{.44\textwidth}
          \centering
          \includegraphics[width = 0.85 \textwidth]{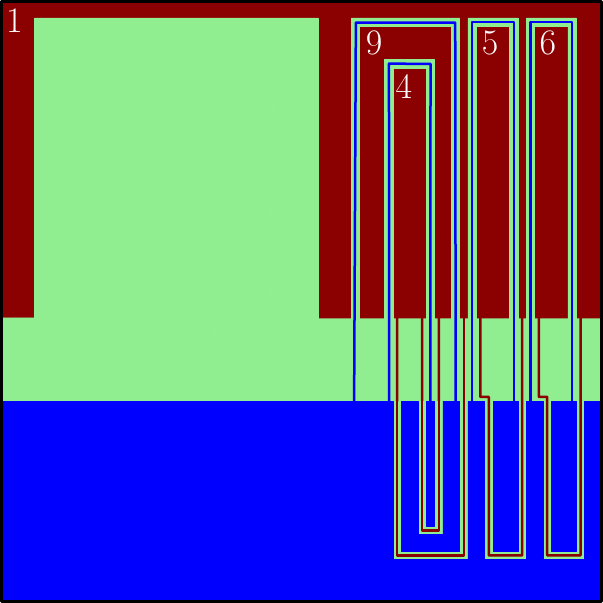}
    \end{minipage}
    \centering
    \begin{minipage}{.54\textwidth}
          \centering
          \includegraphics[width = 0.44 \textwidth]{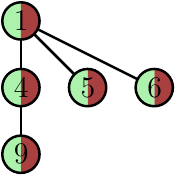}
    \end{minipage}
    \caption{The setup for an exchange sequence. Left: a map returned by a sequence of three gravity moves. Right: the corridor graph $T(P)$ that encodes the topology of $P=D_1\cup D_2$. Every node corresponds to a polygon in $P$. In this example, each polygon intersects both districts, which is indicated by a split node-coloring.}
    \label{fig:afterexchangegravity}
\end{figure}

\begin{restatable}{lemma}{lemexchange}
    \label{lem:exchange}
    Let $M=\{D_1,D_2,D_3\}$ be a 3-district map with the ordering property, and $M'=\{D_1',D_2',D_3'\}$ the map returned by an $\exchange$ sequence on $M$.
    Let $P = D_1 \cup D_2$ and $P' = D'_1 \cup D'_2$.
    Then, $M'$ satisfies the ordering property, and $|E(T(P'))|\leq \frac23 |E(T(P))|$.
\end{restatable}

\begin{proof}
By Lemma~\ref{lem:3split}, there exists a contiguous subtree $T^*$ of $T(P)$ rooted at a centroid $c$ that contains at least one third of the vertices and whose weight, excluding the weight of $c$, is less than $\frac{W}{2}$, where $W = \area(D_1)+\area(D_2)$.

The first ReCom move of the exchange sequence fills all polygons corresponding to the nodes in $T^*$ with green, since its area (excluding the weight of $c$) is less than $\area(D_2)=\max_{i\in \{1,2,3\}}\area(D_i)\geq{W}/2$. 
After this move, the red district is simply connected by Lemma~\ref{lem:gravity_on_disk} and green is connected because it occupies $T(P)$ minus the red district and contains the boundary of $T(P)$.
Note that each edge in $T^*$ is now a single green corridor.

The second ReCom move eliminates edges of $T^*$. We now show that the green district remains connected.
Although this move eliminates all green corridors at once, we can  model the operation as successively removing a pair of corridors (one green and one blue), and show that green remains connected throughout. While $T^*$ has a leaf $\ell$, remove the corridor incident to it.
This creates precisely one cycle in $T(D_3)$ which we break by removing a blue corridor.
The blue cycle ``coats'' $\ell$, or else it would not be a leaf. 
Then the removed blue corridor is on the boundary of $\ell$.
By construction, the blue district shares a boundary only with the green district.
Thus $\ell$ merges into the boundary of some other node of $T(P)$.
Since green contains the boundary of $T(P)$, green stays connected.


Note that the first two moves do not necessarily maintain the ordering property. The last gravity move between $D_1$ and $D_2$ restores the ordering property, as the region containing the top edge of the unit square can be a part of $T^*$ and become green. 
A gravity move restores the ordering property, since $T(P)$ contains the top edge of the unit square and due to Observation~\ref{lem:gravity_on_disk}. Since every node in $T(P')$ already contains both green and red, this step does not introduce new corridors, although it might extend or shorten some, thus it does not add (compressed) complexity to the map.

Since $T^*$ contains at least one third of the vertices it also contains roughly one third of all edges. Consequently, $T(P')$ will have at most $\frac23$ of the edges of $T(P)$.
\end{proof}

\begin{figure}[tbh]
	\centering
    \begin{minipage}{.48\textwidth}
          \centering
          \includegraphics[width = 0.8 \textwidth]{mapcomplex.pdf}
          \caption*{(a)}
    \end{minipage}
    \begin{minipage}{.48\textwidth}
          \centering
          \includegraphics[width = 0.8 \textwidth]{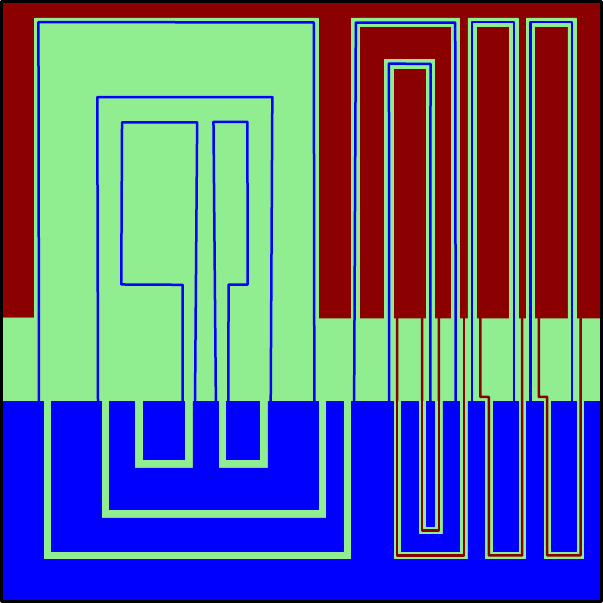}
          \caption*{(b)}
    \end{minipage}
    \begin{minipage}{.48\textwidth}
          \centering
          \includegraphics[width = 0.8 \textwidth]{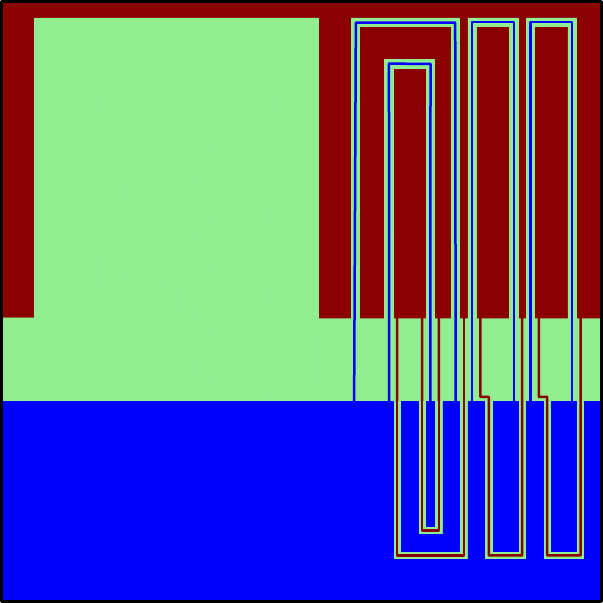}
          \caption*{(c)}
    \end{minipage}
    \begin{minipage}{.48\textwidth}
          \centering
          \includegraphics[width = 0.8 \textwidth]{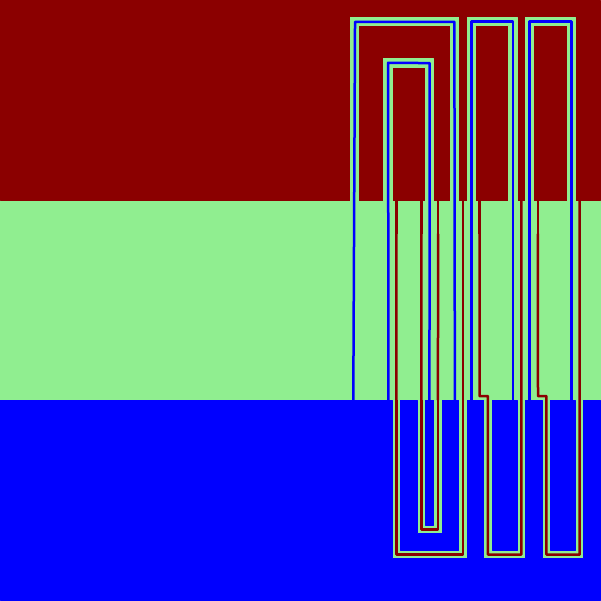}
          \caption*{(d)}
    \end{minipage}\textbf{}
    \caption{An exchange sequence.}
    \label{fig:exchange_ex}
\end{figure}

\subsection{Full Reconfiguration Algorithm}
Overall, the  algorithm for a 3-district map $\map([0,1]^2)=\{D_1,D_2,D_3\}$ works as follows: after a preprocessing phase of $O(1)$ ReCom moves, apply the sequence of three moves $\gravity(D_1, D_2)$, $\gravity(D_2, D_3)$, and $\gravity(D_1, D_2)$; compute the corridor graph $T(P)$ for $P=D_1\cup D_2$; while $T(P)$ has two or more nodes, apply an exchange sequence. Once $T(P)$ has one node, $\gravity(D_1, D_2)$ yields the canonical configuration.

\begin{restatable}{theorem}{thmthree} \label{thm:square3log}
Given a 3-district map $\map([0,1]^2)=\{D_1,D_2,D_3\}$ of complexity $n$, there is a sequence of $O(\log n)$ ReCom moves that transforms it into a canonical map. Furthermore, the districts in each intermediate map are polygons with $O(n)$ vertices and at most one hole.
\end{restatable}


\begin{proof}
    After preprocessing, three \gravity\ moves bring the districts into canonical form with the possible exception of corridors. Each exchange sequence eliminates a constant fraction of corridors by Lemma~\ref{lem:exchange}. After $O(\log n)$ ReCom moves we then obtain the canonical configuration.
    
    The algorithm described above produces a sequence of ReCom moves, where the districts in intermediate maps are weak polygons. 
    By Proposition~\ref{cor:perturbation}, these maps can successively be perturbed into polygons. This completes the proof of the first claim.
    
    It remains to show that the districts in each intermediate map are polygons with $O(n)$ vertices and at most one hole. 
    By construction, the only possible hole appears in the green district after the first step of the \exchange\ sequence. 
    Each of the $O(1)$ ReCom moves in preprocessing adds a corridor with $O(n)$ vertices, and so each district has $O(n)$ vertices at the end of preprocessing.
    By Lemma~\ref{lem:gravitycomplexity}, each gravity move increases the number of vertices by a constant factor. After three gravity moves, each district still has $O(n)$ vertices.
    
    The algorithm applies $O(\log n)$ exchange sequences. At the end of every exchange sequence, the districts are in canonical form with the exception of corridors. Each exchange sequence removes some of the corridors, and does not create new corridors. 
    It should be clear that the complexity of the blue district only decreases since corridors are only eliminated and never created.
    Note that intermediate ReCom moves within an \exchange\ sequence (step 1) may add $O(n)$ new vertices to the red district.
    In an exchange sequence, the 1st ReCom move is a $\gravity$ move w.r.t.\ a sub-polygon, and creates only $O(n)$ new vertices by Lemma~\ref{lem:gravitycomplexity}. The 2nd ReCom move eliminates corridors (and the corresponding vertices); and the 3rd ReCom move eliminates any other vertices created in the 1st move of the sequence. Thus, the complexities of the red and green districts decrease after one \exchange\ sequence.
    
    Finally, when we perturb all weak polygons into polygons in the entire ReCom sequence, the number of vertices remains $O(n)$ for each district by Proposition~\ref{cor:perturbation}.
\end{proof}

\section{Reconfiguration for \texorpdfstring{$k$}{k} Districts}
\label{sec:general}

We generalize our algorithm to an arbitrary number of districts, using recursion. 
For any $3\leq k\leq n$, an \emph{instance} $I = (\map(\mathcal{R}),\map'(\mathcal{R}),\delta)$ of the problem consists of two area-compatible $k$-district maps $\map(\mathcal{R})=\{D_1,\ldots,D_k\}$ and $\map'(\mathcal{R})=\{D'_1,\ldots , D'_k\}$, where $\mathcal{R}$ is a weak polygon with at most one hole, and $\delta$ is a density function.
We define the complexity of $I$ (denoted $|I|$) as the pair $(k,n)$, where $n$ is the maximum over the compressed complexities of $\map$ and $\map'$, and the complexities of all districts $D_i$ and $D_i'$ ($i\in\{1,\ldots,k\}$). 
The overall recursive strategy goes as follows:
First construct a piecewise linear retraction from a punctured unit square $\mathcal{S}$ to $\mathcal{R}$, and extend $\map$ and $\map'$ to two maps on $\mathcal{S}$. 
If $k\geq 4$, then group the $k$ districts into three \emph{superdistricts}, each containing $\lfloor k/3\rfloor$ or $\lceil k/3\rceil$ districts;
and run the algorithm in Section~\ref{sec:three} on the superdistricts.
Note that each ReCom move on a pair of superdistricts is an instance of our problem with fewer districts, which can be solved recursively. The retraction then transforms the ReCom sequence on $\mathcal{S}$ to a ReCom sequence on $\mathcal{R}$. 
We analyze the recursion and give a bound on the number of ReCom moves.

\subsection{Reduction from $\mathcal{R}$ to a Square}
\label{ssec:reduction}

\subparagraph{Filling Holes.}
We now describe how to convert a general instance $I = (\map(\mathcal{R}),\map'(\mathcal{R}),\delta)$ into an instance $I = (\map(\mathcal{S}),\map'(\mathcal{S}),\delta')$ between maps whose domain is a punctured square $\mathcal{S}$. By applying an affine transformation, we may assume that the bounding box of $\mathcal{R}$ is a unit square $[0,1]^2$. We define $\mathcal{S}$ to be $[0,1]^2$ with an interior point $p_H$ removed if $\mathcal{R}$ has a hole. We define a new density function $\delta'$ on $\mathcal{S}$ such that $\delta'(x)=\delta(x)$ for all $x\in \mathcal{R}$ and $\delta'(x)=0$ otherwise. 

We would like to construct $\map(\mathcal{S})$ and $\map'(\mathcal{S})$ simply by giving each connected component of $\mathcal{S}\setminus \mathcal{R}$ to an adjacent district in each map. 
Let $H$ be the weakly simple polygon $\mathcal{S}\setminus \mathcal{R}$  (e.g., a hole of of $\mathcal{R}$, or a region between the outer boundaries of $\mathcal{R}$ and $\mathcal{S}$).
If we fill $H$ with a single district in an arbitrary map $\map(\mathcal{R})$, we may create new adjacencies not originally in $\map(\mathcal{R})$. This happens when $H$ is adjacent to three or more districts. (For example, the magenta district in Fig.~\ref{fig:homotopy}(a) is not adjacent to the red district, but the pair becomes adjacent by filling $H$ with the magenta district as in Fig.~\ref{fig:homotopy}(c).) For this reason, we perform $O(k)$ ReCom moves on $\map(\mathcal{R})$ and $\map'(\mathcal{R}$ to ensure that $H$ is bounded by only one district. 

Chose a district $D_H$ that borders $H$ and successively add corridors on the boundary of $H$ until  $D_H$ becomes the sole district adjacent to $H$. This takes $O(k)$ ReCom moves.

\begin{figure}[htbp]
    \centering
    \includegraphics[width=\textwidth]{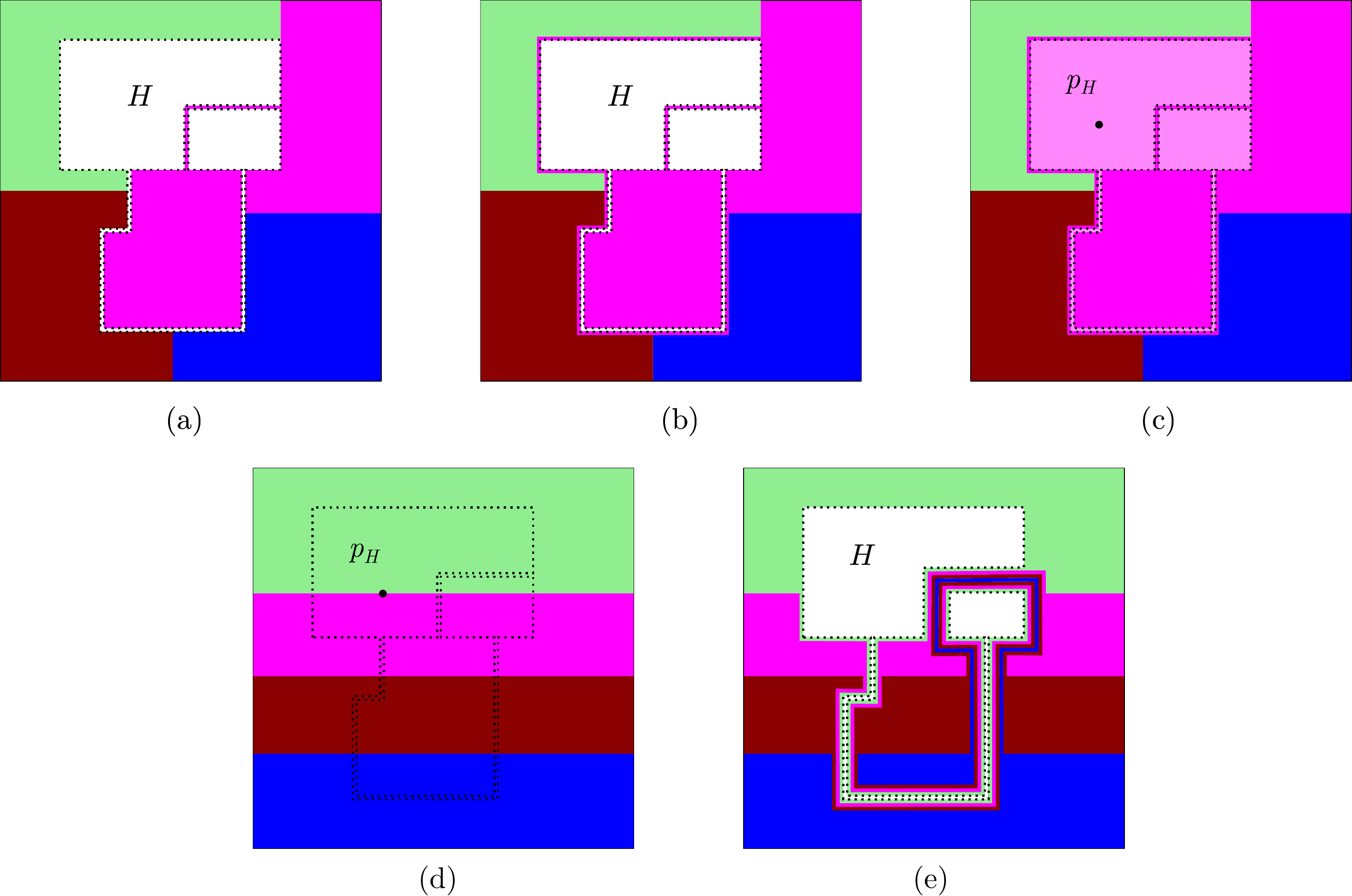}
    \caption{(Top) Transformation of a map in $\mathcal{R}$ (a) into a map in the punctured square $\mathcal{S}$ (c). We first perform $O(k)$ ReCom moves to ``coat'' a weakly simple hole $H$ with a single district (b).
    The retract $\mathcal{H^*}$ maps a weak representation of a map in $\mathcal{S}$ into a weak representation of a map in $\mathcal{R}$.
    }
    \label{fig:homotopy}
\end{figure}

\subparagraph{Retraction.}
Assume that we already have a solution for an instance $I = (\map(\mathcal{S}),\map'(\mathcal{S}),\delta')$ on $\mathcal{S}$, which is a ReCom sequence from $\map(\mathcal{S})$ to $\map'(\mathcal{S})$. 
We show how to transform it into a ReCom sequence from $\map(\mathcal{R})$ to $\map'(\mathcal{R})$ with the same number of ReCom moves. The specific transformation has no impact on the number of ReCom moves, but may increase the complexity of intermediate maps.

We need a continuous function $\mathcal{H^*}$ that transforms a map $\map_i(\mathcal{S})$ into a map $\map_i(\mathcal{R})$ such that $\mathcal{H^*}(\map(\mathcal{S}))=\map(\mathcal{R})$, $\mathcal{H^*}(\map'(\mathcal{S}))=\map'(\mathcal{R})$, and is the identity function on $\mathcal{R}$. Intuitively, $\mathcal{H^*}$ ``compresses'' the regions of $\mathcal{S}\setminus \mathcal{R}$ into $\partial\mathcal{R}$ (it retracts a hole $H$ onto $\partial H$).

We define $\mathcal{P}_\eps(\mathcal{S})$ (resp., $\mathcal{P}_\eps(\mathcal{R})$) as a continuous function that produces a polygon that is an $\eps$-perturbation of $\mathcal{S}$ (resp., $\mathcal{R}$). Specifically, let $\mathcal{P}_\eps(\mathcal{S})$ replace a puncture in $\mathcal{S}$ by a small square.
Let $\partial\mathcal{R}_\eps$ denote the $\eps$-neighborhood of the boundary of $\mathcal{R}$.
Define a homeomorphism $\mathcal{H}$ from $\mathcal{P}_\eps(\mathcal{R})$ to $\mathcal{P}_\eps(\mathcal{S})$ that is the identity when restricted to $\mathcal{R}\setminus \partial\mathcal{R}_\eps$. 
We define $\mathcal{P}_\eps$ such that the function $\mathcal{H^*}$ obtained by $\mathcal{P}_\eps(\mathcal{S})$ composed with $\mathcal{H}_\eps^{-1}$ and $\mathcal{P}_\eps(\mathcal{R})^{-1}$ is a retract, i.e., it is the identity function on  $\mathcal{R}$.

Notice that $\mathcal{H^*}$ defines a map from weak representations of maps in $\mathcal{S}$ to weak representations of maps in $\mathcal{R}$ as desired, where the order of overlapping corridors is given by the homeomorphism $\mathcal{H}$.
Since the boundary of every hole $H$ is ``coated'' by the same district $D_H$ in $\map(\mathcal{R})$ and $H$ is contained in $D_H$ in $\map(\mathcal{S})$, we have that $\mathcal{H^*}(\map(\mathcal{S}))=\map(\mathcal{R})$, as desired, for every homeomorphism $\mathcal{H}$ that meets the constraints described above.
The same applies for $\mathcal{H^*}(\map'(\mathcal{S}))=\map'(\mathcal{R})$.
However, note that in general maps, $\mathcal{H^*}$ might transform a line segment in $\mathcal{S}\setminus\mathcal{R}$ to an arbitrarily long path in $\partial\mathcal{R}$.
Nonetheless, we can describe a ``well-behaved'' piecewise-linear homeomorphism as follows. We consider a piecewise linear homeomorphism from the perturbation of $H$ given by $\mathcal{P}_\eps(\mathcal{R})$ to a unit square such that the center of the square corresponds to the puncture $p_H$. 
We can do this by subdividing the boundary of the square so that it has the same number of vertices as the perturbation of $H$ and triangulating both with compatible triangulations.
There is an obvious piecewise linear retraction of the punctured square $[0,1]^2 \setminus (\frac12,\frac12)$ to the boundary. Translating this retraction via the above homeomorphism gives a retraction of $H$ onto its boundary. The image of any line segment under this retraction has complexity $O(h)\leq O(n)$, where $h$ is the number of vertices in the perturbation of $H$.

\subsection{Reconfiguration for \texorpdfstring{$k$}{k} Districts in a Square}
\label{sec:kInSquare}

\subparagraph{Preprocessing: Ordering Property.}


For ease of notation, from now on we assume that
we are given a $k$-district map $\map(\mathcal{R})=\{D_1,\ldots,D_k\}$ with density $\delta$ of compressed complexity $n$, and that  $\mathcal{R}$ is a punctured square. 
Let $s_1$ and $s_2$ be the left and right sides of $\mathcal{R}$.
Analogously to Section~\ref{sec:three}, we successively append corridors to the districts until every district is simply connected, and intersects both $s_1$ and $s_2$. 
This can be accomplished with $O(k)$ ReCom moves. 
Specifically, while there is a district $D_i$ that is a polygon with holes, let $v$ be a vertex of the hole on its convex hull, and $u$ be a vertex not on the hole that is visible from $v$.
Add $uv$ to the boundary of $D_i$ reducing its total number of holes.
Now all districts are simply connected. While there is a district that is not adjacent to $s_1$ (resp., $s_2$), let $D_i$ be such a district adjacent to some district $D_j$ which is adjacent to $s_1$ ($s_2$); then we recombine $D_i$ and $D_j$ and append to $D_i$ a shortest path to $s_1$ ($s_2$) along the boundary of $D_j$.
Recall that new corridors do not introduce new vertices (they might increase the multiplicity of a vertex), hence the compressed complexity of the map remains $n$. 

After preprocessing, every district has a connected intersection with both $s_1$ and $s_2$; and the counterclockwise order of these intersections along the boundary of $\mathcal{R}$ is the reverse of each other, or else two districts would cross. Therefore the districts can be ordered along $s_1$; assume w.l.o.g.\ that they are labeled $D_1,\ldots , D_k$ in this order.

\subparagraph{Main Algorithm.} 
After preprocessing, the domain $\mathcal{R}$ is a punctured square, and $\map$ and $\map'$ satisfy the ordering property.
In the base case $k=3$, Theorem~\ref{thm:square3log} yields a sequence of $O(\log n)$ ReCom moves that reconfigure
$\map(\mathcal{R})$ to $\map'(\mathcal{R})$.

If $k\geq 4$, then we partition $k$ into a sum of three integers, $k=k_1+k_2+k_3$, where  $k_1,k_2,k_3\in \{\lfloor k/3\rfloor, \lceil k/3\rceil\}$.
Create two 3-district maps, $\tilde{\map}(\mathcal{R})=\{\mathcal{S}_1,\mathcal{S}_2,\mathcal{S}_3\}$ and $\tilde{\map}'(\mathcal{R})=\{\mathcal{S}'_1,\mathcal{S}'_2,\mathcal{S}'_3\}$, where 
$\mathcal{S}_i$ is the union of $k_i$ consecutive districts in $\map(\mathcal{R})$, and 
$\mathcal\mathcal{S}'_i$ is the union of $k_i$ consecutive districts in $\map'(\mathcal{R})$, for $i=1,2,3$.
The union of consecutive districts is simply connected, due to the ordering property, and so $\tilde{\map}(\mathcal{R})$ and $\tilde{\map}'(\mathcal{R})$ are 3-district maps on $\mathcal{R}$.
The maps $\tilde{\map}(\mathcal{R})$ and $\tilde{\map}'(\mathcal{R})$ are area-compatible  with respect to $\delta$, because so are $\map(\mathcal{R})$ and $\map'(\mathcal{R})$. Note also that the compressed complexity of $\tilde{\map}(\mathcal{R})$ and $\tilde{\map}'(\mathcal{R})$ is $O(n)$ since the 
image graph of $\tilde{\map}(\mathcal{R})$ (resp., $\tilde{\map}'(\mathcal{R})$) is a subgraph of the image graph of ${\map}(\mathcal{R})$ (resp., ${\map}'(\mathcal{R})$).
Also note that the complexity of each district $\mathcal{S}_i$ and $\mathcal{S}'_i$, for $i\in\{1,2,3\}$, is $O(n)$ because of the ordering property and the fact that the complexity of the districts in $\map$ and $\map'$ is $O(n)$.

By Theorem~\ref{thm:square3log}, there is a sequence of $O(\log n)$ ReCom moves that reconfigures $\tilde{\map}$ to $\tilde{\map}'$, and all intermediate 3-district maps have complexity $O(n)$. 
The intermediate and final steps are all some 3-district map $\tilde{\map}^{*}(\mathcal{R})=\{\mathcal{S}_1^{*},\mathcal{S}_2^{*},\mathcal{S}_3^{*}\}$ with three superdistricts, which are weak polygons. 
Obtain a $k$-district map ${\map}^{*}$ by subdividing each superdistrict $\mathcal{S}_i^*$ into $k_i$ districts with areas $a_1,\ldots , a_{k_i}$ using waterlines and corridors as follows. 
We describe the construction using $\gravity$ moves, but note that these are not ReCom moves and just a description of each target map ${\map}^{*}$.
We subdivide each superdistrict $\mathcal{S}_i^*=\mathcal{S}^-_i\cup \mathcal{S}^+_i$ with $\area(\mathcal{S}_i^-)=a_1$ and $\area(\mathcal{S}_i^+)=a_2+\ldots +a_{k_i}$ using $\gravity$, and recurse on $\mathcal{S}_i^+$ if necessary.
By Lemma~\ref{lem:gravitycomplexity}, each district has complexity $O(n)$.
 
For each ReCom move in a 3-district map between $\tilde{\map}$ and $\tilde{\map}'$, we construct an instance $I(\ell,m)$ that we solve recursively. Suppose the move recombines superdistricts $\mathcal{S}_i\cup \mathcal{S}_j$ into $\mathcal{S}'_i\cup \mathcal{S}'_j$ within the unit square. 
Let $Q=\mathcal{S}_i\cup \mathcal{S}_j$, which is a connected polygon with $O(n)$ vertices and at most one hole by Theorem~\ref{thm:square3log}. 
The instance $I(\ell,m)$ consists of the two $\ell$-district maps in $Q$ obtained from $\mathcal{S}_i$ and $\mathcal{S}_j$, and $\mathcal{S}_i'$ and $\mathcal{S}_j'$ as explained above. 
Then, $\ell$ is the total number of districts in the two superdistricts, where $2\lfloor k/3\rfloor\leq \ell\leq k-\lfloor k/3\rfloor<k$ by construction, and  $m$ is the maximum compressed complexity between the two maps which is $O(n)$. 



\subparagraph{Analysis of the Number of Moves.} 
We analyse the recurrence of the algorithm 
to obtain the following theorem.

\begin{restatable}{theorem}{thmgeneral}\label{thm:kdistrict}
Given any two area-compatible polygonal $k$-district maps of  complexity at most $n$ in a simply connected domain, $\exp(O(\log k \log \log n))=(\log n)^{O(\log k)}=k^{O(\log\log n)}$ ReCom moves are sufficient to transform one into the other.
Furthermore, the complexity of each map in intermediate steps is $n^{k^{O(1)}}$.
\end{restatable}


\begin{proof}
For $3\leq k\leq n$, let $T(k,n)$ denote the minimum number of ReCom moves that can transform any polygonal $k$-district map to any other with compatible areas, 
and the domain as well as each district is a polygon with at most $n$ vertices. 
From an instance $I(k,n)$, our algorithm makes $O(\log n)$ recursive calls of the form $I(\frac{2k}{3},c\cdot n)$, where $c$ is a constant.
Then,
\[
T(k,n)\leq O\left(T\left(\frac{2k}{3},c\cdot n\right)\cdot \log n\right)+O(k).
\]

The height of the recursion tree is $O(\log k)$ and the maximum branching factor is $O(\log(n\cdot c^{\log k}))=O(\log n + \log k)=O(\log n)$ since $k<n$.
Then $T(k,n)$ solves to $\exp(O(\log k \log \log n))= (\log n)^{O(\log k)} = k^{O(\log\log n)}$.
%
By Proposition~\ref{cor:perturbation}, we can convert the ReCom sequence on weak representation to a ReCom sequence of the same length in which all districts are simple polygons.

The analysis above prioritized the number of ReCom moves, rather than the complexity of the map at intermediate steps.
For instance, consider the recursion that simulates a ReCom move of superdistricts transforming a $k$-district map $M(\mathcal{R})$ into $M'(\mathcal{R})$.
Our algorithm recurses on a $\frac{2k}{3}$-district map of complexity $c\cdot n$ on a punctured square $\mathcal{S}$, which yields a sequence of $O(\log n)$ ReCom moves. However, to convert this into a sequence of ReCom moves on $k$-district maps one must apply the retraction $\mathcal{H}^{*}$ (cf.\ Section~\ref{ssec:reduction}) to every intermediate map, retracting a weakly simple polygon $H$ to its boundary $\partial H$. 
Since the complexity of $H$ could be $\Omega(n)$, $H$ might cross the same district $\Omega(n)$ times, which causes $\mathcal{H}^{*}$ to push the district into $\Omega(n)$ narrow corridors along the boundary of $H$. (In Fig.~\ref{fig:homotopy}(e), two corridors of the red district run parallel along the boundary of $H$.)
This might cause the complexity of the district to increase to $\Omega(n^2)$ in intermediate steps. The retraction $\mathcal{H}^{*}$, described in Section~\ref{ssec:reduction}, ensures that the complexity goes up from $n$ to at most $O(n^2)$ after applying $\mathcal{H}^{*}$ in each recursive step. Since the depth of the recursion tree is $O(\log k)$, the maximum complexity of all intermediate maps is $n^{2^{O(\log k)}}=n^{k^{O(1)}}$.
Note that this does not increase the number of ReCom moves since $M$ and $M'$ are determined in the parent level, and $\mathcal{H}^{*}$ is only applied to recover intermediate steps between $M$ and $M'$, which are obtained from lower complexity maps in the children level.
\end{proof}

\section{Lower Bound Construction}
\label{sec:lower}

\newcommand{\arc}{\gamma}
\renewcommand{\iota}{m}

This section shows that $\Omega(\log n)$ ReCom moves are sometimes necessary to transform a given map of complexity $n$ into canonical form, even for three districts of equal areas in $[0,1]^2$.

\subparagraph{Overview.}
We describe an initial map with 3 districts in a unit square, and show that after $k$ ReCom moves, each district contains an arc of a specific combinatorial pattern (defined below). These arcs are defined recursively, each iteration roughly tripling the complexity of the arcs.
Thus the total complexity of the arcs in iteration $\ell$ is $O(3^\ell)$.
The initial district map is a thickening of one of these arcs after $\iota\geq 6$ iterations. 
We show that if each district contains an arc from iteration $\ell$, then after a recombination they each contain an arc of iteration $\ell-4$.
In the canonical configuration, each district can only contain arcs of iteration $1$. Then, the number of recombinations from the initial district map to the canonical configuration is at least linear in the number of iterations.

\subparagraph{Construction.} 
We first describe the family of simple arcs $F_{\ell}$, for all $\ell \in \mathbb{N}_0$, mentioned in the overview.
All arcs in $F_{\ell}$ will start at the $\varepsilon$-neighborhood of the left side of the square and end at the $\varepsilon$-neighborhood of the right side, crossing the middle section $3^{\ell}$ times.
Each family $F_{\ell}$ can be described with a combinatorial pattern, namely, the order in which the arcs traverse the $3^{\ell}$ segments in the middle section of the square.
In the base case, $F_0$ is the set of arcs that cross the middle section only once.
Given an arc $\arc_{\ell} \in F_{\ell}$, we describe an arc $\arc_{\ell+1} \in F_{\ell+1}$.
We construct two arcs, $\arc_{\ell}^+,\arc_{\ell}^-\in F_{\ell}$, that closely follow $\arc_{\ell}$ on the left and on the right, respectively, and are mutually noncrossing. Then $\arc_{\ell+1}$ is the concatenation of $\arc_{\ell}$, the reverse of $\arc_{\ell}^-$, and $\arc_{\ell}^+$, where two consecutive arcs are connected by short arcs in the left and right $\varepsilon$-neighborhoods of the square; see Fig~\ref{fig:fractalbase}.
Let $F_{\ell+1}$ be the family of all arcs with the same combinatorial pattern as $\arc_{\ell+1}$.
The following observation follows by construction.

\begin{figure}[tbh]
	\centering
    \begin{minipage}{.27\textwidth}
          \centering
          \includegraphics[width = 0.95 \textwidth]{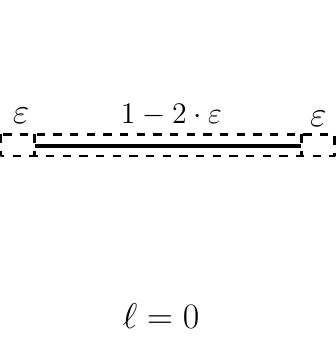}
    \end{minipage}
    \begin{minipage}{.32\textwidth}
          \centering
          \includegraphics[width = 0.94 \textwidth]{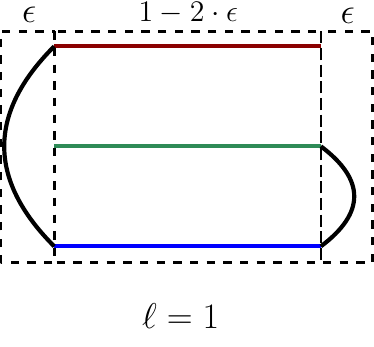}
    \end{minipage}
    \begin{minipage}{.35\textwidth}
          \centering
          \includegraphics[width = 0.95 \textwidth]{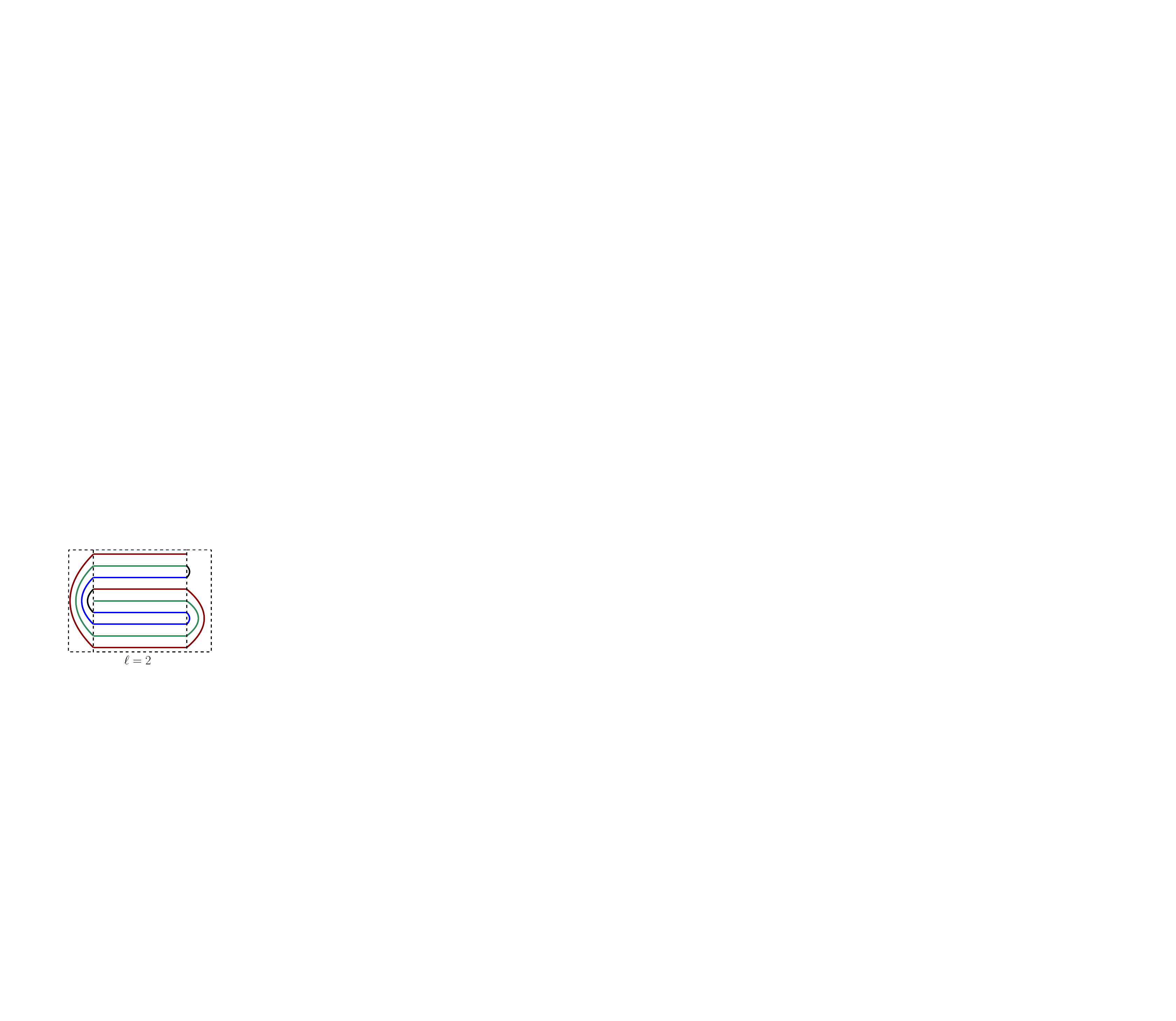}
    \end{minipage}
    \caption{The first three levels of the recursive construction for arcs in $F_\ell$, for $\ell \in \{ 0, 1, 2 \}$. Note that the blue, green, and red arcs for  $\ell = 2$ each resemble a copy of the entire stage for $\ell=1$.}
    \label{fig:fractalbase}
\end{figure}

\begin{observation}\label{obs:recursion}
For  $0\leq j\leq \ell$, we can partition every arc $\arc_\ell\in F_{\ell}$  into $3^j$ arcs in $F_{\ell-j}$.
\end{observation}

\subparagraph{Initial Map.}
We call each of the $3^j$ arcs described in Observation~\ref{obs:recursion} a \emph{$(\ell-j)$-anchor} of $\arc_{\ell}$.
The initial map is drawn relative to an arc $\arc_\iota \in F_\iota$ whose middle segments are equally spaced horizontal line segments in the unit square. 
We call $\arc_\iota$ the \emph{initial arc}.
The map is a ``thickening'' of $\arc_\iota$, defined as follows. 
The middle section is partitioned into $3^\iota$ rectangles of width $1-2\varepsilon$ and equal area.
Each of the three districts is created based on one of the three $(\iota-1)$-anchors of $\arc_\iota$.
We use the portions of the anchors of $\arc_\iota$ in the $\varepsilon$-neighborhoods of the vertical sides of the unit square to construct corridors that make each district connected.

In  the  remainder of this section, we show that after $i$ ReCom moves, each district contains an arc in $F_{\iota-4i-1}$ that also satisfies additional conditions used to ensure the induction.
We split the conditions into what we call a spacing invariant and an anchoring invariant.

\subparagraph{Spacing Invariant.} 
An arc $\arc_\ell\in F_\ell$ satisfies the \emph{spacing invariant} if there exists a partition of the middle section of the unit square into $3^\ell$ \textbf{uncertainty regions} and $3^\ell+1$ \textbf{gap regions}, which are axis-aligned rectangles of width $1-2\varepsilon$ such that the uncertainty (resp., gap) regions are pairwise interior disjoint, and they each have height at most $\alpha_\ell = 3^{-(\iota+\ell+1)/2}$ (resp., at most $g_\ell=2(3^{-\ell}-3^{-\iota})$).
Each middle segment in $\arc_\ell$ is contained in a unique uncertainty region, i.e., there is a one-to-one relation between middle segments and uncertainty regions.

\subparagraph{Anchoring Invariant.}
Recall that the initial map splits the middle section of the square into rectangles of equal height.
An arc $\arc_\ell\in F_\ell$, for $0\leq \ell<\iota$, satisfies the anchoring invariant if it satisfies 
the following conditions.
There exists an $\ell$-anchor $A_\ell$ of the initial arc $M_\iota$ such that 
each uncertainty region of $\arc_\ell$ contains a rectangle created by the thickening of a middle segment of $A_\ell$, and $M_\iota$ traverses the uncertainty regions in the same order as $A_\ell$ traverses the corresponding rectangles.

Let $S_\ell\subset F_\ell$ be the family of arcs in $F_\ell$ that satisfy the spacing and anchoring invariants.

\begin{restatable}{lemma}{lemrecursion}\label{lem:recursion}
For every $1\leq \ell\leq \iota$, every $\ell$-anchor of the initial arc $\arc_\iota$ is in $S_{\ell}$.
\end{restatable}

\begin{proof}
The claim for $\ell=\iota$ is trivial since $\arc_\iota\in S_{\iota}$ by making each rectangle an uncertainty region and the gap regions being degenerate rectangles with $0$ height.
Using Observation~\ref{obs:recursion}, we can produce a hierarchical decomposition of $\arc_\iota$ into $3^{\iota-\ell}$ anchors in $F_{\ell}$.
This decomposition forms a ternary tree whose root is $\arc_\iota$ and whose leaves are the $\ell$-anchors.
Let $\arc_{\ell}$ be one of the leaves and $\arc_{\ell+1}$ be its parent arc.
(For example, we can see each of the three colored arcs in Fig.~\ref{fig:fractalbase}(c) as $\arc_{\ell}$ and the entire arc as the parent $\arc_{\ell+1}$.)
Recall that the initial map splits the middle section of the square into rectangles of equal height.
That is, each segment gets thickened into a rectangle of height $3^{-\iota}$.
We use the rectangles containing unique segments of $\arc_{\ell}$ as the uncertainty regions, which satisfies $3^{-\iota}\le \alpha_\ell$.
Thus the gap regions are the regions of the middle section of the unit square minus the uncertainty regions.
It remains to show that every gap region is a rectangle of height at most $g_\ell$.

By construction, each segment in the middle section of $\arc_{\ell}$ gets tripled in the parent curve $\arc_{\ell+1}$ (three segments consecutive in the vertical order, each in one of the children of $\arc_{\ell+1}$).
Thus, a gap between two consecutive segments (in the vertical order) in $\arc_{\ell}$ contains 
at most four segments of $\arc_{\ell+1}$ (which happens for the bottom-most red segments in Fig.~\ref{fig:fractalbase}(c)).
By recursion, each of these segments in the gap triple and we get at most $4$ additional segments from the pair of segments of $\arc_{\ell}$ on the boundary of the gap.
Thus, the number of segments of $\arc_\iota$ in a gap of $\arc_{\ell}$ is bounded above by the geometric series $4(1+3+3^2+\ldots+3^{\iota-\ell-1})=4\frac{3^{\iota-\ell}-1}{2}= 2(3^{\iota-\ell}-1)$.
Since each segment is contained in a rectangle of height $3^{-\iota}$, the height of a gap region is at most $g_\ell=2(3^{-\ell}-3^{-\iota})$, as required.
The anchoring invariant is trivially satisfied since $\arc_{\ell}$ is an anchor.
\end{proof}


\begin{corollary}\label{cor:basecase}
In the initial map, each district contains an arc in $S_{\iota-1}$.    
\end{corollary}

\subparagraph{Dual Graphs.}
For every arc $\arc_\ell\in F_{\ell}$, we define the \emph{dual graph} $G(\arc_\ell)$ based on the complement of $[0,1]^2 \setminus \arc_\ell$ as follows.
The vertices of $G(\arc_\ell)$ are the maximal connected regions of the middle section of $[0,1]^2 \setminus \arc_\ell$ between two segments of $\arc_\ell$, and one additional vertex for the \emph{outer face}, which is the union of the two components in the middle section  adjacent to the top and bottom sides of the unit square, and the components of the complement in the left and right $\varepsilon$-neighborhoods that contain the left and right sides of the unit square. Two vertices of $G(\arc_\ell)$ are adjacent if their corresponding pair of regions are adjacent through a connected component of the complement in an $\varepsilon$-neighborhoods of the left or right side of the square; see Figure~\ref{fig:curvedual}.
The next two lemmas establish properties of $G(\arc_\ell)$ that are used later.

\begin{figure}[tbh]
	\centering
	\includegraphics[width = \textwidth]{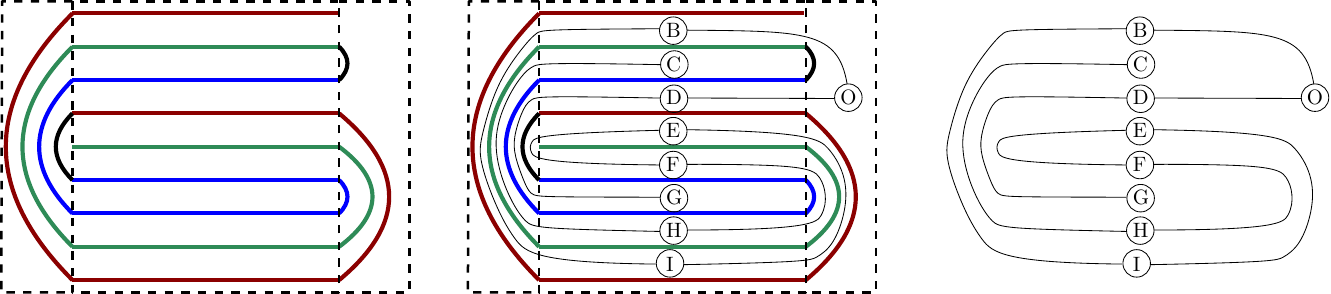}
    \caption{Left: an arc $\arc_2 \in F_2$. Middle: the arc and its dual graph overlaid. Notice that the two paths from the node $O$ have length $2$ and $6$. Right: the dual graph $G(\arc_2)$.}
    \label{fig:curvedual}
\end{figure}

\begin{restatable}{lemma}{lempaths}
\label{lem:paths}
For all $\ell \ge 1$, the dual graph $G(\arc_\ell)$ is a tree obtained by attaching $\ell$ disjoint paths $p_1, p_2,\ldots, p_\ell$ to a root $O$, where the length of each path $p_i$ is $2\cdot3^{i-1}$.
\end{restatable}

\begin{proof}
We prove this by induction on $\ell$. In the base case, when $\ell = 1$, there is a single path of length $2$, as seen in Figure~\ref{fig:fractalbase} (b).
For the inductive step, notice that each path in $G(\arc_\ell)$ is also present in $G(\arc_{\ell+1})$, as a part of the copy of $G(\arc_\ell)$ we started with.
Recall that $\arc_{\ell+1}$ extends $\arc_\ell$ by adding an arc that follows $\arc_\ell$ along both sides.
This triples the number of vertices in the dual and introduces a single new path in $G(\arc_{\ell+1})$.
Thus, the length of the new path is $2 \cdot 3^\ell$.
\end{proof}

\begin{restatable}{lemma}{lemalignment}
    \label{lem:alingment}
    Every path $p_i$, defined in Lemma~\ref{lem:paths}, can be partitioned into two paths $p_i^-$ and $p_i^+$ with $3^{i-1}$ vertices, each of them is immediately between two $(i-1)$-anchors of $\arc_\ell$.
\end{restatable}

\begin{proof}
    We proceed by induction on $\ell$. The base case is trivial when $\ell=1$ and $G(\arc_\ell)$ is a path of length $2$.
    Recall that we add the arcs $\arc_\ell^-$ and $\arc_\ell^+$ to $\arc_\ell$ to obtain $\arc_{\ell+1}$, and that $\arc_\ell^-$, $\arc_\ell^+$ and $\arc_\ell$ are $\ell$-anchors of $\arc_{\ell+1}$.
    Thus, 
    $G(\arc_{\ell+1})$ can be obtained from $G(\arc_\ell)$ by adding a path $p_{\ell+1}$ of length $2\cdot 3^\ell$.
    Note that $p_{\ell+1}^+$ follows $\arc_\ell^+$ immediately between $\arc_\ell^+$ and $\arc_\ell$ while $p_{\ell+1}^-$ is immediately between $\arc_\ell$ and $\arc_\ell^-$.
    (In Fig.~\ref{fig:curvedual}, for example, the path $(B,I,E)$ is between the red and green curves and $(F,H,C)$ is between the green and blue curves.)
    The paths $p_{\ell}^+$ and $p_{\ell}^-$ are sandwiched between two $(\ell-1)$-anchors of $\arc_\ell$, by the induction hypothesis, and now are sandwiched between two $(\ell-1)$-anchors corresponding to $\arc_\ell^+$ and $\arc_\ell^-$. The remaining paths are not affected.
\end{proof}

\begin{lemma}\label{lem:area}
If $1\leq \ell \leq \iota-4$, then the total area of the uncertainty regions of any arc $\gamma_\ell \in S_\ell$ is at most $\frac{1}{9}$. 
\end{lemma}
\begin{proof}
For every arc $\arc_\ell\in S_\ell$, there are $3^\ell$ uncertainty regions, each of which a rectangle of width $1-2\varepsilon$ and height $\alpha_\ell = 3^{-(\iota+\ell)/2}$. 
Thus their total area is 
$3^\ell\cdot (1-2\varepsilon) 3^{-(\iota+\ell)/2} < 3^{(\ell-\iota)/2}\leq 3^{-2}$ when $\ell \leq \iota-4$.
\end{proof}

Within the middle section (of area $1-2\varepsilon$), the complement of all uncertainty regions is precisely the union of gap regions. Consequently, Lemma~\ref{lem:area} implies the following. 

\begin{corollary}\label{cor:area}
For all $\varepsilon<\frac{1}{18}$, $\iota \geq 6$, and $0\le  \ell\leq \iota-4$, the total area of the gap regions is at least $\frac{7}{9}$.
\end{corollary}

We now focus on a ``volume'' argument to show that each district must occupy a significant portion of $G(\arc_\ell)$ after a ReCom move.
Note that for all arcs $\arc_\ell\in F_\ell$, the dual graphs $G(\arc_\ell)$ are isomorphic.
For every arc $\arc_\ell\in S_{\ell}$, we define a weight function $w:V(G(\arc_\ell))\rightarrow \mathbb{R}^+$ as follows. By the spacing invariant, the region corresponding to each vertex $v$ is contained in the union of a gap region (a rectangle of height at most $g_\ell$), and up to two uncertainly regions that contain the segments of $\arc_\ell$ 
(rectangles of height $\alpha_\ell$). The weight $w(v)$ of a vertex $v$ is the area of intersection of the corresponding region and the gap region (excluding the area in the uncertainty regions). Consequently, for all $v\in V(G(\arc_\ell))$, it is true that:
\begin{equation} \label{eq:gap-upper}
w(v)
\leq (1-2\varepsilon)g_\ell
\leq (1-2\varepsilon)\cdot 2(3^{-\ell}-3^{-\iota})
\le 2\cdot 3^{-\ell}
\end{equation}

\subparagraph{Special Paths.}
We now define subpaths of $G(\arc_\ell)$ that will allow us to quantify 
the amount of ``progress'' that one ReCom move can make towards the canonical configuration. 
We focus on the three longest paths $p_{\ell-2}$, $p_{\ell-1}$ and $p_{\ell}$ as defined in Lemma~\ref{lem:paths}.
We call these three paths the \emph{long paths} of $G(\arc_\ell)$ and the remaining paths the \emph{short paths} of $G(\arc_\ell)$.
We assume that $\ell\ge 6$. 
We partition the vertices in each long path (excluding the root) into vertex-disjoint \emph{special paths} of length $3^{\ell-4}$.
Note that, by Lemma~\ref{lem:paths}, the paths $p_{\ell-2}$, $p_{\ell-1}$ and $p_{\ell}$ contain $6$, $18$ and $54$ special paths, respectively.
Thus, in total, there are $78$ special paths in $G(\arc_\ell)$.
We say that a district \emph{visits} a special path $p^*$ if the district intersects (the region corresponding to) each node in  $p^*$.
We use \eqref{eq:gap-upper} to show the following lemma.

\begin{restatable}{lemma}{lemthreelongest}\label{lem:3longest}
Given a 3-district map where one district contains an arc $\arc_\ell\in S_\ell$, then each of the other two districts visits at least one special path of the dual graph $G(\arc_\ell)$. 
\end{restatable}

\begin{proof}
    By Lemma~\ref{lem:paths}, the number of vertices in the short paths is $\sum_{i=1}^{\ell-3}|p_i|=2\frac{3^{\ell-4}-1}{2}=3^{\ell-4}-1$, where $|p_i|=2\cdot 3^{i-1}$ is the length of path $p_i$.
    By equation \eqref{eq:gap-upper}, the weight of each node is at most $2\cdot 3^{-\ell}$, and the total weight these nodes is at most $2\cdot3^{-\ell}(3^{\ell-4}-1)=2(3^{-4}-3^{-\ell})<\frac{2}{81}$.
    By Corollary~\ref{cor:area}, the 3 long paths account for an area of at least $\frac{7}{9}-\frac{2}{81}=\frac{61}{81}$ of the total area.

    Each district has $1/3$ area, at most $1/9$ of which may be in uncertainty regions, and so it has at least $2/9$ area in gap regions.
    We distinguish between two cases:
    
    \emph{Case~1: The district visits the root of $G(\arc_\ell)$.}
    Since $G(\arc_\ell)$ has at most $\frac{2}{81}$ area in its short paths, at least $\frac29-\frac{2}{81}=\frac{16}{81}$ must be in the 3 long paths.
    Then each district visits at least $(\frac{16}{81})/(2\cdot 3^{-\ell})= \frac{8}{81}\cdot 3^\ell$ nodes of the long paths.
    The long paths jointly have $2(3^{\ell-1}+3^{\ell-2}+3^{\ell-3})=\frac{26}{27}\cdot 3^{\ell}$ nodes.
    Then each district is present in $\frac{8}{81}/\frac{26}{27}=\frac{4}{39}$ portion of nodes in the 3 longest paths.
    The case that minimizes the maximum number of nodes visited in a long path is when the visited nodes are equally divided among the 3 long paths. 
    Then, each path has $\frac{4}{3\cdot 39}<\frac{1}{81}$ nodes visited by the district. Thus the district visits one of the special paths. 
    
    \emph{Case~2: The district does not visit the root of $G(\arc_\ell)$.} 
    Then the district is contained in one of the three long paths. 
    Furthermore, it cannot visit any of the short paths.
    Thus, the district must visit at least $(\frac{2}{9})/(2\cdot 3^{-\ell})= 3^{\ell-2}$ nodes.
    That is the length of $9$ special paths.
    Thus the district visits at least $8$ special paths, since $3^{\ell-4}$ of the visited vertices can be in special paths that do not have all its vertices visited by the district.
\end{proof}

\begin{lemma}\label{lem:key}
Let $\arc_\ell\in S_{\ell}$ be an arc contained in a district and $\arc$ be an arc visiting the sequence of gap regions of $\arc_\ell$ corresponding to a special path $p^*$ in $G(\arc_\ell)$. Then $\arc \in S_{\ell-4}$.
\end{lemma}

\begin{proof}
    Refer to Figure~\ref{fig:lower-bound-regions}.
    By combining Observation~\ref{obs:recursion}, Lemma~\ref{lem:alingment} and the anchoring invariant, we conclude that each special path is between two $(\ell-4)$-anchors of $\arc_\iota$ used as anchors for $\arc_\ell$.
    We use the union between the gap region containing a segment of $\arc$ and the up to two adjacent uncertainty regions of $\arc_\ell$ as the uncertainty region of $\arc$.
    Then the gap regions of $\arc$ is the complement.
    Note that that includes the anchors sandwiching $p^*$ and, therefore, the uncertainty regions of $\arc$ contain an $(\ell-4)$-anchor of $\arc_\iota$.
    By Lemma~\ref{lem:recursion}, each of these anchors has a gap region with height at most $g_{\ell-4}$.
    Then the gap between two vertically consecutive uncertainty regions of $\arc$ is at most by $g_{\ell-4}$.
    It remains to show that the height of the uncertainty regions of $\arc$ is at most $\alpha_{\ell-4}$.
    Recall that, by definition, $\alpha_\ell = 3^{-(\iota+\ell+1)/2}$ and $g_\ell=2(3^{-\ell}-3^{-\iota})$.
    Since $\arc_\ell\in S_{\ell}$, the height of the uncertainty regions of $\arc$ is at most $2\alpha_\ell+g_\ell$.
    By simplifying the inequality $2\alpha_\ell+g_\ell \le \alpha_{\ell-4}$ we obtain $3^{\frac{\ell}{2}} \le 80 \cdot3^\frac{\iota+1}{2}$, which holds for $\ell\le \iota$.
\end{proof}

\begin{figure}[h]
    \centering
    \includegraphics[width=.7\textwidth]{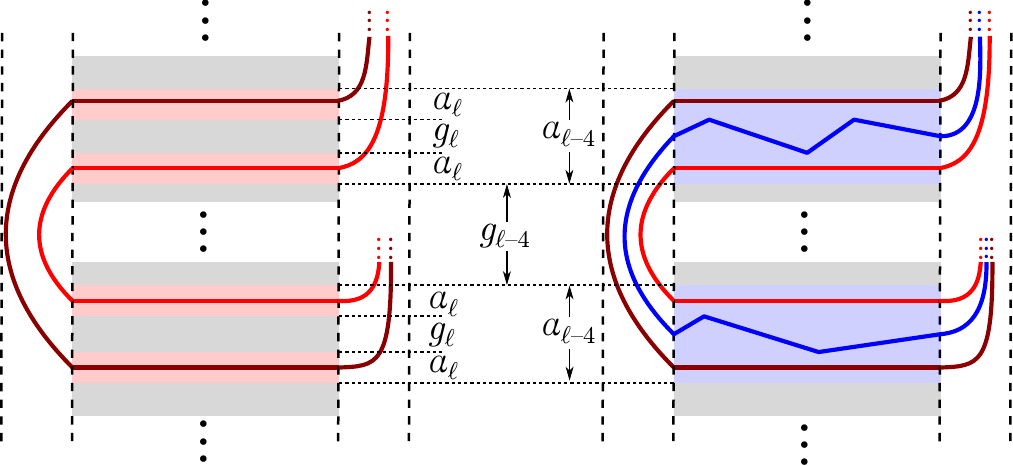}
    \caption{Illustration of Lemma~\ref{lem:key}. The red (resp., blue) rectangles are the uncertainty regions of $\arc_\ell$ (resp., $\arc$). Two $(\ell-4)$-anchors of $\arc_\iota$ are shown as light and dark red arcs, and $\arc$ is shown as a blue arc.}
    \label{fig:lower-bound-regions}
\end{figure}

\begin{lemma}\label{lem:induction}
Given the initial map defined above, after $i$ ReCom moves, for $0\leq i\leq \frac{\iota}{4}-1$, each district contains an arc in $S_{\iota-4i-1}$.
\end{lemma}
\begin{proof}
We proceed by induction on $i$. In the base case, $i=0$, the claim follows from Corollary~\ref{cor:basecase}.
In the induction step, assume that the claim holds for some $i\in \mathbb{N}_0$, and then an $(i+1)$st ReCom move is performed. The \emph{obstacle district} is the district that is not modified by the $(i+1)$st ReCom move, and w.l.o.g. is $D_3$. By the induction hypothesis, $D_3$ contains an arc $\arc_\ell\in S_\ell$ for $\ell=\iota-4i-1$. 
By Lemma~\ref{lem:recursion}, a subarc of $\arc_\ell$ is in $S_{\ell-4}$, as required.
It remains to find arcs in $S_{\ell-4}$ in the two districts that have been recombined, namely $D_1$ and $D_2$.
By Lemma~\ref{lem:3longest}, both $D_1$ and $D_2$ visit a special path of $G(\arc_\ell)$ in any district map.
For each of these two districts, we can then find an arc satisfying the conditions of Lemma~\ref{lem:key}.
Then, by Lemma~\ref{lem:key}, $D_1$ and $D_2$ each contain an arc in $S_{\iota-4i-5}$, as required.
% ...
\end{proof}

Lemma~\ref{lem:key} and \ref{lem:induction} show that the districts that visit a special path have arcs in $S_{\ell-4}$ and, therefore, each ReCom move can only decrease the ``level'' of arcs contained in the district by at most four. Thus, we need $\Omega(\iota)$ ReCom moves to get the canonical configuration and, by making $\iota=\log n$ we get the following theorem.

\begin{restatable}{theorem}{thmlower}
    \label{thm:lower-bound}
    There exist two area-compatible 3-district maps, $\map$ and $\map'$, both with complexity $O(n)$, such that $\Omega(\log n)$ ReCom moves are necessary to reconfigure $\map$ into $\map'$, even when the districts in both maps are axis-aligned orthogonal polygons with vertices on an integer grid of size $O(n)\times O(n)$.
\end{restatable}

\begin{proof}
    The initial map $\map$ is a thickening of the arc $\arc_\iota$. 
    The complexity of $\map$ is $3^{O(\iota)}$, thus we make $\iota=\log n$.
    Let $M'$ be the canonical configuration, defined by the thickening of a curve $\arc_1\in S_1$ as in the definition of the initial map, i.e., all 3 district are congruent rectangles.
    Note that the choice of a rational $\eps<\frac{1}{18}$ allows us scale the maps by $O(n)$ so that we can describe the corridors in the $\eps$-neighborhood of the vertical sides of the domain with orthogonal curves with axis aligned edges lying on the integer grid. 
    By Lemma~\ref{lem:induction}, after any sequence of  $\frac{\iota-2}{4}=\Theta(\log n)$ ReCom moves, each district in the resulting map still contains an arc in $S_1$.
    By the spacing invariant, the districts in $\map'$ cannot contain an arc in $S_j$ with $j\ge 1$ because the gap regions require that a vertical line through the center of the map intersect each district more than once.
    Thus, the length of any sequence of ReCom moves that reconfigures $\map$ into $\map'$ is $\Omega(\log n)$.
\end{proof}

\section{Conclusions}
\label{sec:con}

We have shown that (in our continuous setting) any pair of area-compatible district maps can be reconfigured into each other by a sequence of area-preserving recombination moves. Though the discrete version of this result remains unsolved (see Related Work), our result suggests that for any two maps, with a discretization of the geographic domain which is granular enough, we can connect them by ReCom moves. However, establishing quantitative bounds on the necessary granularity is left for future work.

Between 3-district maps, the number of recombination moves is $O(\log n)$, where $n$ is the combinatorial complexity of the maps, matching our worst-case lower bound of $\Omega(\log n)$. Between $k$-district maps, for $k\geq 4$, we construct a sequence of $\exp(O(\log k\log \log n))=(\log n)^{O(\log k)}$ ReCom moves. It remains an open problem whether the number of moves can be reduced to be polynomial in both $k$ and $n$. For $k\geq 4$ districts, our algorithm uses a recursion of depth $O(\log k)$. However, this approach increases the complexity of intermediate maps to $n^{k^{O(1)}}$. It is also an open problem whether there exists a sequence of ReCom moves where the complexity of intermediate maps remains polynomial in both $k$ and $n$. 

\bibliography{bib}

\end{document}